%% file: main.tex
\documentclass[acmsmall,nonacm]{acmart}

\allowdisplaybreaks

\input{macros}

\begin{document}

\title{Acyclic Conjunctive Regular Path Queries are no Harder than Corresponding Conjunctive Queries}

\author{Mahmoud Abo Khamis}
\orcid{0000-0003-3894-6494}
\email{mahmoudabo@gmail.com}
\affiliation{
  \institution{\rai{}}
  \city{Berkeley}
  \state{CA}
  \country{USA}
}

\author{Alexandru-Mihai Hurjui}
\orcid{0009-0009-2799-4694}
\email{hurjuialexandru12@gmail.com}
\affiliation{%
    \institution{University of Zurich}
    \city{Zurich}
    \country{Switzerland}
}

\author{Ahmet Kara}
\orcid{0000-0001-8155-8070}
\email{ahmet.kara@oth-regensburg.de}
\affiliation{
  \institution{OTH Regensburg}
  \city{Regensburg}
  \country{Germany}
}

\author{Dan Olteanu}
\orcid{0000-0002-4682-7068}
\email{olteanu@ifi.uzh.ch}
\affiliation{%
    \institution{University of Zurich}
    \city{Zurich}
    \country{Switzerland}
}

\author{Dan Suciu}
\orcid{0000-0002-4144-0868}
\email{suciu@cs.washington.edu}
\affiliation{
  \institution{University of Washington}
  \city{Seattle}
  \state{WA}
  \country{USA}
}

\begin{abstract}
    \input{abstract}
\end{abstract}



\maketitle

\input{intro}

\input{prelims}
\input{results}
\input{related}
\input{conclusion}

\begin{acks}
Suciu's work was partially supported by NSF IIS 2314527, NSF SHF 2312195, NSF III 2507117, and a Microsoft professorship. Olteanu's work is partially supported by SNSF 200021-231956.
\end{acks}

\bibliographystyle{ACM-Reference-Format}
\bibliography{bibliography}

\appendix
\input{app_prelims}
\input{app_results}
\input{app_runtime_comparison}

\end{document}

%% file: macros.tex
\usepackage{bm} 
\usepackage{enumitem}
\usepackage{subcaption}
\usepackage{xspace}
\usepackage{algorithm}
\usepackage[noend]{algpseudocode}
\usepackage{tikz}
\usetikzlibrary{calc}
\usetikzlibrary{shapes,decorations.markings,arrows.meta,fit}
\usetikzlibrary{decorations.pathreplacing}
\usetikzlibrary{automata,positioning}
\usepackage{accents}

\usepackage{thmtools, thm-restate}

\newcommand{\defeq}{\stackrel{\text{def}}{=}}
\newcommand{\set}[1]{\{#1\}}                    
\newcommand{\setof}[2]{\{{#1}\mid{#2}\}}        
\newcommand{\bag}[1]{\{\!\!\{#1\}\!\!\}}

\theoremstyle{definition}
\newtheorem{theorem}{Theorem}[section]   
\newtheorem{definition}[theorem]{Definition}
\newtheorem{remark}[theorem]{Remark}
\newtheorem{claim}[theorem]{Claim}

\usepackage[colorinlistoftodos]{todonotes}

\newtheoremstyle{cited}%
{.5\baselineskip\@plus.2\baselineskip
    \@minus.2\baselineskip}
{.5\baselineskip\@plus.2\baselineskip
    \@minus.2\baselineskip}
{\itshape}
{\parindent}
{}
{.}
{.5em}
{\textsc{\thmname{#1}} \thmnote{\normalfont#3}}

\newcommand{\cq}{\text{CQ}\xspace}
\newcommand{\cqs}{\text{CQs}\xspace}
\newcommand{\rpq}{\text{RPQ}\xspace}
\newcommand{\rpqs}{\text{RPQs}\xspace}
\newcommand{\crpq}{\text{CRPQ}\xspace}
\newcommand{\crpqs}{\text{CRPQs}\xspace}
\newcommand{\freeleaf}{\text{free-leaf}\xspace}

\newcommand{\boundconnected}{\text{bound-connected}\xspace}
\newcommand{\Boundconnected}{\text{Bound-connected}\xspace}
\newcommand{\out}{\textsf{OUT}}

\definecolor{light-gray}{gray}{0.7.2}
\definecolor{goodgreen}{rgb}{0.1, 0.5, 0.1}
\definecolor{burntorange}{rgb}{0.8, 0.33, 0.0}
\definecolor{lightblue}{RGB}{173, 216, 230} %
\colorlet{DarkGreen}{green!50!black}
\colorlet{lightred}{red!30!white}
\colorlet{lightgreen}{green!30!white}
\colorlet{lightbrown}{brown!30!white}

\newcommand{\calG}{\mathcal{G}}

\newcommand{\calE}{\mathcal{E}}

\newcommand{\calV}{\mathcal{V}}

\newcommand{\nop}[1]{}

\newcommand{\change}[1]{#1}

\newcommand{\rai}{RelationalAI}

\newcommand{\ospg}{{\sf{OSPG}}\xspace}
\newcommand{\osyan}{{\sf{OSYan}}\xspace}
\newcommand{\pc}{{\sf{PC}}\xspace}

\newcommand{\vars}{\textsf{vars}}
\newcommand{\atoms}{\textsf{atoms}}
\newcommand{\free}{\textsf{free}}
\newcommand{\bound}{\textsf{bound}}
\newcommand{\nodes}{\textsf{nodes}}
\newcommand{\fhtw}{\textsf{fhtw}}
\newcommand{\ghtw}{\textsf{ghtw}}
\newcommand{\td}{\textsf{TD}}
\newcommand{\fnfhtw}{\textsf{fn-fhtw}}
\newcommand{\fnghtw}{\textsf{fn-ghtw}}
\newcommand{\ftd}{\textsf{FTD}}

\newcommand{\mylist}{\textsf{list}}
\newcommand{\dom}{\textsf{Dom}}
\newcommand{\eliminate}{\textsf{eliminate}}

\newcommand{\cw}{\mathsf{cw}}

\newcommand{\ov}{\overline}

\newcommand{\regpath}[1]{\accentset{\rightsquigarrow}{#1}}

%% file: abstract.tex
We present an output-sensitive algorithm for evaluating an acyclic Conjunctive Regular Path Query (CRPQ). Its complexity is written in terms of the input size, the output size, and a well-known parameter of the query that is called the ``free-connex fractional hypertree width''. Our algorithm improves upon the complexity of the recently introduced output-sensitive algorithm for acyclic CRPQs~\cite{crpq-icdt26}. More notably, the complexity of our algorithm for a given acyclic CRPQ $Q$ matches the best known output-sensitive complexity for the ``corresponding'' conjunctive query (CQ), that is the CQ that has the same structure as the CRPQ $Q$ except that each RPQ is replaced with a binary atom (or a join of two binary atoms). This implies that it is not possible to improve upon our complexity for acyclic CRPQs without improving the state-of-the-art on output-sensitive evaluation for acyclic CQs. Our result is surprising because RPQs, and by extension CRPQs, are equivalent to recursive Datalog programs, which are generally poorly understood from a complexity standpoint. Yet, our result implies that the recursion aspect of acyclic CRPQs does not add any extra complexity on top of the corresponding (non-recursive) CQs, at least as far as output-sensitive analysis is concerned.

%% file: intro.tex
\section{Introduction}
\label{sec:intro}
\change{
Graph databases have become increasingly popular in recent years, with applications in various domains such as knowledge graphs, social networks, and the semantic web.
These graph databases are often queried differently from traditional relational databases, using graph query languages such as SPARQL~\cite{SPARQL:2013} and Cypher~\cite{Cypher:SIGMOD:2018}.
Recursion is at the heart of graph query languages, and it is often expressed using regular expressions over edge labels, leading to the notion of Regular Path Queries (RPQs) and their conjunctive extensions, Conjunctive Regular Path Queries (CRPQs).
CRPQs are part of SPARQL, the
W3C standard for querying RDF data, including commonly used
knowledge bases such as DBpedia and Wikidata~\cite{FigueiraGKMNT20,MartensT19b}.
Recent analytical studies of SPARQL query logs indicate that the use of RPQs and CRPQs is on the rise~\cite{10.14778/3149193.3149196}. For example, RPQs currently make up 38\% of unique queries to Wikidata~\cite{10.1145/3308558.3313472,10.1007/978-3-030-00668-6_23}.
}

\change{In this paper, we study acyclic \crpqs in the data complexity setting, i.e.,
where the query is {\em fixed}.}
We provide an {\em output-sensitive} algorithm for evaluating acyclic \crpqs,
i.e., an algorithm whose complexity depends on both the input size and output size.
We show that our algorithm improves upon prior algorithms for acyclic \crpqs~\cite{crpq-icdt26,rpq-pods25}.
We also show that our complexity matches the best known output-sensitive complexity
for the corresponding acyclic conjunctive queries~\cite{xiao-os-yannakakis}.
Hence, it is not possible to improve upon our complexity for acyclic \crpqs~
{\em without} a new breakthrough on output-sensitive evaluation of acyclic conjunctive queries.

Given an alphabet $\Sigma$, an \emph{edge-labeled graph} $G = (V, E, \Sigma)$
is a directed graph with a vertex set $V$ and an edge set $E$ where each edge is labeled
with a symbol from $\Sigma$.
The size of $G$ is $|V| + |E|$ and we denote it by $N$ throughout the paper.
Given an edge-labeled graph $G=(V, E, \Sigma)$,
a \emph{regular path query (\rpq)} is a regular expression over $\Sigma$.
We denote a regular expression as $\regpath{R}$.
The output of an RPQ $\regpath{R}$ on $G$ is the set of all pairs of vertices $(v, u) \in V \times V$ where $G$ contains a path from $v$ to $u$ whose label (i.e., the concatenation of the labels of its edges) belongs to the language defined by $\regpath{R}$.
Recent work~\cite{rpq-pods25} introduced an \emph{output-sensitive} algorithm for evaluating RPQs, meaning that its complexity depends not just on the size $N$  of the input graph but also on the output size $\out$. Its complexity is $O(N + N \cdot \out^{1/2})$.\footnote{The complexity was stated in~\cite{rpq-pods25} as $O(N + N^{3/2} + \out\cdot N^{1/2})$.
However, upon examining the analysis from~\cite{rpq-pods25}, it turns out that the complexity
can be expressed as $O(N + N \cdot\Delta + N \cdot \out / \Delta)$ for any $\Delta > 0$.
To minimize this expression, we set $\Delta = \out^{1/2}$, yielding the stated complexity.}

Given an edge-labeled graph $G=(V, E, \Sigma)$,
a {\em Conjunctive Regular Path Query (\crpq)} $Q$ is a conjunction of atoms
followed by a projection onto a set of {\em free variables}, where each atom is by itself an \rpq.
In particular, each atom has the form $\regpath{R}(X, Y)$ where $\regpath{R}$ is some regular expression over $\Sigma$ and $X$ and $Y$ are two variables.
For example, consider the following two \crpqs:
\begin{align}
    Q_1(X, Y_1, Y_2, Y_3) &=
        \regpath{R_1}(X, Y_1) \wedge
        \regpath{R_2}(X, Y_2) \wedge
        \regpath{R_3}(X, Y_3)
        \label{eq:intro-q1}\\
    Q_2(Y_1, Y_2, Y_3) &=
        \regpath{R_1}(X, Y_1) \wedge
        \regpath{R_2}(X, Y_2) \wedge
        \regpath{R_3}(X, Y_3)
        \label{eq:intro-q2}
\end{align}
The first query $Q_1$ is asking for tuples of vertices $(x, y_1, y_2, y_3)$ in $G$
such that\footnote{For a natural number $n$, we use $[n]$ to denote $\{1,2,\ldots,n\}$. When $n=0$, $[n]=\emptyset$.} for every $i \in [3]$, there is a path from $x$ to $y_i$ whose label belongs to the language defined by $\regpath{R_i}$.
The second query $Q_2$ is asking for triples of vertices $(y_1, y_2, y_3)$
such that there exists a vertex $x$ where for every $i \in [3]$, there is a path from $x$ to $y_i$ whose label belongs to the language defined by $\regpath{R_i}$.
In this paper, we are only interested in \crpqs that are {\em acyclic}, which means
that the query structure is a forest.\footnote{We will see later that the notion of acyclicity for \crpqs is slightly different from acyclicity for CQs, in the sense
that acyclic \crpqs cannot have {\em self-loops} or {\em parallel-edges}. See~Definition~\ref{defn:acyclic-crpq} and Remark~\ref{rem:acyclic_crpq} for details.}
Both $Q_1$ and $Q_2$ above are acyclic since they both have a query structure of a 3-star.
Acyclic \crpqs can be {\em free-connex}, which is equivalent\footnote{While this definition might seem different from the standard definition of free-connex acyclic \cqs~\cite{BaganDG07}, the two definitions coincide for connected acyclic \crpqs because all atoms are binary. See Section~\ref{subsec:width-measures}.} to saying that the free variables
are connected. This is true for $Q_1$ but not for $Q_2$ since, in the latter, the free variables
$Y_1, Y_2, Y_3$ are not connected.

Recent work~\cite{crpq-icdt26} has introduced an output-sensitive algorithm for evaluating acyclic \crpqs.
At a high-level, you can think of their algorithm as ``lifting'' the Yannakakis algorithm~\cite{Yannakakis81} from acyclic conjunctive queries (CQs) to acyclic \crpqs.
Let's take the query $Q_1$ above.
We cannot obtain an output-sensitive algorithm for $Q_1$ by simply evaluating each RPQ $\regpath{R_i}(X, Y_i)$ independently using the algorithm from~\cite{rpq-pods25} and then joining their outputs.
This is because the output of each RPQ $\regpath{R_i}(X, Y_i)$ could be much larger than the final output of $Q_1$ itself.
Instead, the main insight from~\cite{crpq-icdt26} is to start by ``calibrating''
the \rpqs $\regpath{R_i}(X, Y_i)$ with each other before invoking the algorithm
from~\cite{rpq-pods25} on them and joining their outputs.
The calibration filters out vertices that don't contribute to the final output, similar to the calibration passes in the original Yannakakis algorithm~\cite{Yannakakis81}.
Using this insight, the authors from~\cite{crpq-icdt26} show how to evaluate any acyclic
\crpq that is {\em free-connex} in time $O(N + N\cdot \out^{1/2} + \out)$, where $\out$ is the output size.
This applies to the example query $Q_1$ above.

\change{Extending the approach from~\cite{crpq-icdt26} from free-connex \crpqs to acyclic \crpqs that are {\em not} free-connex turns out to be more challenging.
Such non-free-connex queries are transformed in~\cite{crpq-icdt26} into free-connex ones by ``promoting'' some non-free variables to become free.}
For example, $Q_2$ above can be transformed into a free-connex query by promoting $X$ to become free, thus obtaining $Q_1$.
Using this approach, the authors of~\cite{crpq-icdt26} obtain a ``weak'' output-sensitive
algorithm for $Q_2$ whose runtime is $O(N + N \cdot \out_1^{1/2} + \out_1)$
where $\out_1$ is the output size of the {\em transformed query} $Q_1$, which could be much larger than the original output size $\out_2$ of $Q_2$.
To rewrite the runtime in terms of the original output size $\out_2$,
the authors use a weak upper bound of $\out_1 \leq \out_2 \cdot |V|$,
which gives a runtime of $O(N + N \cdot \out_2^{1/2} \cdot |V|^{1/2} + \out_2 \cdot |V|)$.
This is unsatisfactory because it introduces a new parameter $|V|$ into the runtime,
is still a loose bound for many input instances,
and doesn't address non-free-connex queries at a fundamental level.
Instead, it over-approximates them using blackbox algorithms that were not designed for them.
In the worst case, $|V|$ could be as large as $N$, making the runtime
$O(N + N^{3/2} \cdot \out_2^{1/2} + N \cdot \out_2)$,
which can be significantly improved as we explain below.

Instead, in this paper, we develop a custom-made approach for acyclic \crpqs, both free-connex
and non-free-connex, where the concept of ``free-connex'' naturally occurs at the heart of our approach.
Its complexity depends on $N$, $\out$, and a parameter of the query $Q$,
which has been known in the literature on CQs for many years under different names~\cite{DBLP:journals/tods/OlteanuZ15,faq},
but was recently referred to as the {\em free-connex fractional hypertree width}, of the query $Q$~\cite{xiao-os-yannakakis}, denoted $\fnfhtw(Q)$.\footnote{In particular, the free-connex fractional hypertree width is just like the fractional hypertree width except that we only consider {\em free-connex} tree decompositions~\cite{BaganDG07}.
See Sec.~\ref{subsec:width-measures}.}
The free-connex fractional hypertree width is a measure of how far a query is from being free-connex.
It is traditionally defined for \cqs, but we extend it to \crpqs in the natural way.
Its value is 1 when the query is free-connex, and it increases as the query becomes ``less'' free-connex.
For example, the query $Q_1$ above has $\fnfhtw(Q_1) = 1$ since it is free-connex,
while $Q_2$ has $\fnfhtw(Q_2) = 3$.
Our main result says that any acyclic \crpq $Q$ can be evaluated in the following time, in data complexity:
\begin{align}
    O(N + N\cdot \out^{1-\frac{1}{\max(\fnfhtw(Q), 2)}}+\out)
    \label{eq:intro-runtime}
\end{align}
The above complexity collapses back to the complexity from~\cite{rpq-pods25} for \rpqs and to the complexity from~\cite{crpq-icdt26} for free-connex acyclic \crpqs.
For non-free-connex acyclic \crpqs, our complexity is significantly better than the one from~\cite{crpq-icdt26}.
For example, our complexity for $Q_2$ becomes $O(N + N \cdot \out_2^{2/3} + \out_2)$, where $\out_2$ is the output size of $Q_2$.
This is a more natural runtime expression and is also {\em always}\footnote{To compare the two runtimes, note that for $Q_2$, we always have $\out_2 \leq N^3$.} upper bounded by the runtime from~\cite{crpq-icdt26}, which was
$O(N + N^{3/2} \cdot \out_2^{1/2} + N \cdot \out_2)$.
See Appendix~\ref{app:comparison} for a detailed comparison with~\cite{crpq-icdt26}.

But now we ask: How good is our runtime bound from Eq.~\eqref{eq:intro-runtime}?
Is there a way to establish a lower bound showing how optimal it is?
To answer this question, we use conjunctive queries as ``points of reference'' to establish
{\em lower bounds} on \crpqs.
In particular, each \crpq $Q$ must be at least as hard to evaluate as the ``corresponding'' CQ $Q'$, where each RPQ atom $\regpath{R}(X, Y)$ in $Q$ is replaced with a binary relation atom $S(X, Y)$ in $Q'$ of size $N$. For example, the \crpq $Q_2$ above is at least as hard to evaluate as the following CQ:
\begin{align*}
    Q'_2(Y_1, Y_2, Y_3) &=
        S_1(X, Y_1) \wedge
        S_2(X, Y_2) \wedge
        S_3(X, Y_3)
\end{align*}
This is because given relation instances for $S_1, S_2, S_3$ for $Q_2'$,
we can construct a corresponding edge-labeled graph $G$ over an alphabet $\Sigma = \{S_1, S_2, S_3\}$ where each tuple $S_i(x, y)$ corresponds to an edge from $x$ to $y$ labeled with $S_i$.
And now, evaluating $Q_2'$ reduces to evaluating $Q_2$ on $G$ where each $\regpath{R_i}$
is chosen to be a single symbol $S_i$.
In fact, a \crpq $Q$ could be even harder.
For example, even a single \rpq $\regpath{R}$ is at least as hard to evaluate
as the following $k$-path CQ, for every $k \geq 1$:
\begin{align}
    Q_{k\text{-path}}(X_1, X_{k+1}) &= S_1(X_1, X_2) \wedge S_2(X_2, X_3) \wedge \cdots \wedge S_k(X_k, X_{k+1})
    \label{eq:intro-k-path}
\end{align}
This is because the regular expression $\regpath{R}$ can be chosen to be the concatenation of $k$ symbols $S_1 \cdots S_k$.
For a given acyclic \cq $Q'$, the best known output-sensitive {\em combinatorial algorithm}\footnote{{\em Combinatorial algorithms} is an under-defined concept
in the algorithms literature that basically refers to algorithms that don't use fast matrix multiplication, e.g. Strassen's algorithm~\cite{Strassen1969GaussianEI}.} has the following complexity~\cite{xiao-os-yannakakis,paris-os-yannakakis}:
\begin{align}
    O(N + N\cdot \out^{1-\frac{1}{\fnfhtw(Q')}}+\out)
    \label{eq:intro-cq-runtime}
\end{align}
For example, for $Q_2'$ above, the best known combinatorial algorithm runs in time $O(N + N \cdot \out^{2/3} + \out)$,
which happens to match the runtime of our algorithm for $Q_2$, which is also combinatorial.
But this  means that we cannot improve our complexity for the \crpq $Q_2$
{\em without} improving the complexity for the CQ $Q_2'$ from~\cite{DBLP:conf/icdt/AmossenP09,xiao-os-yannakakis}.
More generally, our runtime expression for an acyclic \crpq $Q$ from Eq.~\eqref{eq:intro-runtime} matches
the runtime expression for the corresponding $\cq$ $Q'$ from Eq.~\eqref{eq:intro-cq-runtime},
whenever the free-connex fractional hypertree width is at least two.
Moreover, when the free-connex fractional hypertree width is one\footnote{The free-connex fractional hypertree width for acyclic queries takes only {\em integer} values~\cite{xiao-os-yannakakis}. See Proposition~\ref{prop:fn-fhtw-properties}.},
our runtime expression becomes $O(N + N\cdot \out^{1/2} + \out)$,
which matches the runtime of the best known combinatorial algorithm\footnote{In particular, when $k \geq 2$, we have $\fnfhtw(Q_{k\text{-path}}) = 2$, hence the runtime from Eq.~\eqref{eq:intro-cq-runtime} becomes $O(N + N\cdot \out^{1/2} + \out)$.} for the CQ $Q_{k\text{-path}}$ above for $k \geq 2$~\cite{xiao-os-yannakakis}.
This means that improving our runtime bound any further can {\em only} happen if a new breakthrough is made
on the CQ side.
In other words, our output-sensitive algorithm for acyclic \crpqs is
the best possible, as far as we understand \cq evaluation today.


%% file: prelims.tex
\section{Preliminaries}
\label{sec:prelims}

\subsection{Languages and Graphs}
An {\em alphabet} $\Sigma$ is a finite set of symbols.
The set of all strings over an alphabet $\Sigma$ is denoted by $\Sigma^*$.
The set of all strings over $\Sigma$ of length $k$ is denoted by $\Sigma^k$.
The empty string is denoted by $\epsilon$.
A {\em language} $L$ over an alphabet $\Sigma$ is a subset of $\Sigma^*$.
Given a string $w \in \Sigma^*$, the {\em reverse} of $w$, denoted  $w^R$, is defined recursively as follows: 
$w^R = \epsilon$ if $w = \epsilon$ and 
$w^R = \sigma \cdot v^R$ if $w = v\cdot \sigma$ for some $v \in \Sigma^*$
and $\sigma \in \Sigma$.
The reverse of a language $L$ is $L^R = \{w^R \mid w \in L\}$.
We use the standard definitions of regular expressions and regular languages over an alphabet $\Sigma$.
Throughout the paper, we denote regular expressions as $\regpath{R}, \regpath{S}$, etc. Given a regular expression $\regpath{R}$, we use $L(\regpath{R})$ to denote the language defined by $\regpath{R}$.

\begin{definition}[Edge-Labeled Graph]
An edge-labeled graph $G = (V, E, \Sigma)$ consists of a finite set $V$ of vertices, a finite set $\Sigma$ of edge labels, and a set
    $E \subseteq V \times \Sigma \times V$ of
    labeled directed edges. Every triple $(v, \sigma, u)$ in $E$
    denotes an edge from vertex $v$ to vertex $u$ labeled with $\sigma$.
    The size of $G$ is $|V| + |E|$, i.e., the total number of its vertices and edges.\footnote{We neglect the size of $\Sigma$ by assuming that it does not contain any symbol that does not appear as an edge label in $G$.}
\end{definition}

\begin{definition}[A Path in an Edge-Labeled Graph]
    Given an edge-labeled graph $G = (V, E, \Sigma)$ and two vertices $v, u\in V$, a {\em path} $p$ from $v$ to $u$ of length $k$ for some
    natural number $k$ is a sequence of $k+1$ vertices
    $w_0 := v, w_1, \ldots, w_{k} := u$ and a sequence of $k$ edge labels $\sigma_1, \sigma_2, \ldots, \sigma_k$
    such that for every $i \in [k]$, there is an edge $(w_{i-1}, \sigma_i, w_{i}) \in E$.
    The {\em label} of the path $p$, denoted by $\sigma(p)$, is the string $\sigma_1 \sigma_2 \ldots \sigma_k \in \Sigma^k$.
\end{definition}

\begin{definition}[Transpose of an Edge-Labeled Graph]
    Given an edge-labeled graph $G = (V, E, \Sigma)$,
    let $\Sigma^{-1}$ be a set of fresh symbols:
    one distinct fresh symbol $\sigma^{-1}$ for each symbol $\sigma \in \Sigma$.
    The {\em transpose} of $G$ is the edge-labeled graph $G^T = (V, E^T, \Sigma^{-1})$, where
    $E^T \defeq \{(u, \sigma^{-1}, v) \mid (v, \sigma, u) \in E\}$.
    \label{defn:transpose-graph}
\end{definition}

\begin{definition}[Inverse of a Regular Expression]
    Given a regular expression $\regpath{R}$ over an alphabet $\Sigma$,
    let $\Sigma^{-1}$ be a set of fresh symbols:
    one distinct fresh symbol $\sigma^{-1}$ for each symbol $\sigma \in \Sigma$.
    The {\em inverse} of $\regpath{R}$, denoted $\regpath{R}^{-1}$, is a regular expression over $\Sigma^{-1}$ that defines the reverse language
    of $L(\regpath{R})$ but where every symbol $\sigma$ is replaced by $\sigma^{-1}$.
    In particular, $\regpath{R}^{-1}$ is defined inductively as follows:
    \begin{align*}
        \regpath{R}^{-1} \defeq
        \begin{cases}
            \epsilon & \text{if } \regpath{R} = \epsilon, \\
            \sigma^{-1} & \text{if } \regpath{R} = \sigma \in \Sigma,\\
            (\regpath{R}_1^{-1} + \regpath{R}_2^{-1}) & \text{if } \regpath{R} = \regpath{R}_1 + \regpath{R}_2, \\
            (\regpath{R}_2^{-1} \cdot \regpath{R}_1^{-1}) & \text{if } \regpath{R} = \regpath{R}_1 \cdot \regpath{R}_2, \\
            (\regpath{R}_1^{-1})^* & \text{if } \regpath{R} = (\regpath{R}_1)^*.
        \end{cases}
    \end{align*}
    \label{defn:inverse-regex}
\end{definition}

\begin{definition}[Product Graph, $G_1 \times G_2$]
\label{defn:product-graph}
Given two edge-labeled graphs $G_1 = (V_1, E_1, \Sigma)$ and $G_2 = (V_2, E_2, \Sigma)$,
the {\em product graph} $G_1 \times G_2$ is a directed graph
    $G=(V, E)$ (without edge labels) where $V \defeq V_1 \times V_2$ and
        $E \defeq \{((v_1, v_2), (u_1, u_2)) \mid
            \exists \sigma \in \Sigma : (v_1, \sigma, u_1) \in E_1 \wedge (v_2, \sigma, u_2) \in E_2\}$.
\end{definition}

\subsection{Queries}

A Conjunctive Query (CQ) has the form
$Q(\bm F) = R_1(\bm X_1) \wedge \cdots \wedge 
    R_n(\bm X_n)$,
where $\vars(Q)\defeq \bigcup_{i \in [n]} \bm X_i$ is the set of variables,
$\free(Q) \defeq \bm F \subseteq \vars(Q)$ is the set of {\em free} variables,
$\bound(Q) \defeq \vars(Q)\setminus\free(Q)$ is the set of {\em bound} variables, 
and $\atoms(Q)$ is the set of atoms $R_i(\bm X_i)$.
We follow the standard semantics for CQs.
The query graph of a CQ $Q$ is a hypergraph $\calG = (\calV, \calE)$ where $\calV := \vars(Q)$ and
$\calE := \set{\bm X_i \mid R_i(\bm X_i) \in \atoms(Q)}$.
\change{We call a CQ $Q$ {\em acyclic} if its query graph is $\alpha$-acyclic,
which means that we can construct a tree whose nodes are the sets $\bm X_i$ for $i \in [n]$
where each variable $X\in\bm V$ appears in a connected subtree.}
\begin{definition}[Conjunctive Regular Path Queries]
A {\em Conjunctive Regular Path Query (\crpq)} $Q$ over an alphabet 
$\Sigma$ is of the form
$Q(\bm F) = \regpath{R}_1(X_1, Y_1) \wedge \cdots \wedge 
    \regpath{R}_n(X_n, Y_n)$,
where: each $\regpath{R}_i$ is a \change{regular expression over $\Sigma$}; 
each $\regpath{R}_i(X_i, Y_i)$ is an {\em atom};
$\vars(Q) \defeq \bigcup_{i \in [n]} \{X_i, Y_i\}$ is the set of variables;
$\free(Q) \defeq \bm F \subseteq \vars(Q)$ is the set of {\em free} variables;
$\bound(Q) \defeq \vars(Q)\setminus\free(Q)$ is the set of {\em bound} variables; 
and $\atoms(Q)$ is the set of atoms.
The answer to the \crpq $Q$ on an input edge-labeled graph $G = (V, E, \Sigma)$ is defined as follows:
A mapping $\mu :\free(Q) \rightarrow V$ is in the answer of $Q$ if it can be extended to a mapping $\mu' : \free(Q) \cup \bound(Q) \rightarrow V$ such
that $G$ has a path from $\mu'(X_i)$ to $\mu'(Y_i)$ labeled by a string from $L(\regpath{R}_i)$, for every $i \in [n]$.
We represent a mapping $\mu$ by the tuple $(\mu(X))_{X \in \free(Q)}$,
assuming a fixed order on the free variables.
\end{definition}

\paragraph{Query Graphs and Acyclic \crpqs}
In order to define the query graph of a \crpq, we first recall the definition of undirected multigraphs.
An {\em undirected multigraph} $\calG = (\calV, \calE)$ consists of a finite set of vertices $\calV$ and a multiset  of edges $\calE$ such that each edge is a subset of $\calV$ of size two or one. 
An edge of size one is called a {\em self-loop}.  
The {\em degree} of a vertex $v$  in $\calG$ is the number of edges
incident to $v$.
A vertex that has degree one or zero is called a {\em leaf}.
Given two vertices
$v, u\in \calV$, a {\em path} from $v$ to $u$ in $\calG$ is a sequence 
of vertices $w_0 := v, w_1, \ldots, w_{k} := u$ for some $k \geq 1$
    such that for every $i \in [k]$, there is an edge $\{w_{i-1}, w_{i}\} \in \calE$. 

\begin{definition}[Query Graph of a \crpq]
The {\em query graph} of a \crpq $Q$ is an undirected multigraph
$\calG=(\calV,\calE)$ with vertex set $\calV :=\vars(Q)$ and
 edge multiset $\calE := \bag{\{X, Y\} \mid \regpath{R}(X, Y) \in \atoms(Q)}$.
\end{definition}

A {\em cycle} in an undirected multigraph $\calG=(\calV, \calE)$ is a sequence of the form
$v_1,e_1,$ $v_2,e_2,$ $\ldots,$ $v_k,e_k, v_{k+1}:= v_1$ of $k$ distinct vertices $v_1, \ldots, v_k\in\calV$ and $k$ distinct edges $e_1, \ldots, e_k\in\calE$ such that $k\geq 1$ and for every $i \in [k]$, we have $e_i=\{v_i, v_{i+1}\}$.
A multigraph is called {\em acyclic} if it has no cycles.

\begin{definition}[Acyclic \crpqs]
    \label{defn:acyclic-crpq}
A \crpq $Q$ is called {\em acyclic} if its query graph is acyclic.
\end{definition}
\begin{remark}[Distinction between acyclicity in \crpqs versus CQs]
\label{rem:acyclic_crpq}
    Note that according to the above definition, an acyclic \crpq cannot have a {\em self-loop}, i.e., an atom of the form $\regpath{R}(X, X)$, or {\em parallel edges}, i.e., two atoms sharing the same variables.
    This is more restricted than the definition of ($\alpha$-)acyclicity for \cqs in the literature,
    where a \cq with an atom $R(X, X)$ or two atoms $R_1(X, Y), R_2(X, Y)$ can still be acyclic.
    However, this restriction is unavoidable because a \crpq with a self-loop, $\regpath{R}(X, X)$,
    is already as hard as evaluating the following $k$-cycle conjunctive query for any $k \geq 3$:
    \begin{align}
        Q_{k\text{-cycle}}(X_1) &= S_1(X_1, X_2) \wedge S_2(X_2, X_3) \wedge \cdots \wedge S_k(X_k, X_{1})
    \label{eq:k-cycle}
    \end{align}
    This is because the regular expression $\regpath{R}$ could be chosen to be a concatenation
    of $k$ symbols, $S_1 \cdots S_k$.
    A similar reasoning applies to parallel edges,
    where a \crpq consisting of only two parallel edges is already as hard as $Q_{k\text{-cycle}}$ for any $k \geq 3$.
\end{remark}


In light of the above discussion it is possible to define acyclicity for \crpqs in terms of their {\em $k$-expansions} to \cqs:
\begin{definition}[$k$-expansion of a \crpq]
    Given a \crpq $Q$ and a natural number $k \geq 1$,
    the {\em $k$-expansion} of $Q$, is a \cq $Q_k'$ that is obtained from $Q$
    by replacing every atom $\regpath{R}(X, Y)$ in $Q$ with a $k$-path of binary atoms
    $S(X, Z_1), S(Z_1, Z_2), \ldots, S(Z_{k-1}, Y)$, where $Z_1, \ldots, Z_{k-1}$ are {\em fresh} variables and $S$ is an arbitrary relation symbol.
    \label{defn:k-expansion}
\end{definition}
\begin{proposition}
    A \crpq $Q$ is acyclic if and only if its $3$-expansion is an acyclic \cq.
\end{proposition}
The reason we use $3$-expansions above (and not just 2-expansions) is because a self-loop will only result in a cyclic \cq after a $3$-expansion.

A \crpq $Q$ is called {\em connected}  if its query graph is connected.  A variable in $Q$ is called a {\em leaf} if it is a leaf in the query graph
of $Q$. A query $Q'$ is a subquery of $Q$ if $\atoms(Q') \subseteq \atoms(Q)$ 
and $\free(Q') = \free(Q) \cap \vars(Q')$. 

We state some basic observations on evaluating \rpqs and \crpqs.

\begin{restatable}{proposition}{PropLinearRPQ}\label{prop:linear-rpq}
    (The \crpq $Q(X) = \regpath{R}(X, Y)$ is evaluatable in linear time.)
    Given an edge-labeled graph of size $N$, the \crpq $Q(X) = \regpath{R}(X, Y)$ can be evaluated in $O(N)$ time.
\end{restatable}

The following proposition follows immediately from Definitions~\ref{defn:transpose-graph} and~\ref{defn:inverse-regex}. It basically says that to produce the output of an \rpq $\regpath{R}(Y, X)$
where $X$ and $Y$ are swapped, we can simply take the {\em inverse} of $\regpath{R}$ and evaluate it on the {\em transpose graph}.
\begin{restatable}{proposition}{PropRPQTranspose}
    \label{prop:rpq-transpose}
(The \crpq $Q(X, Y) = \regpath{R}(Y, X)$ is equivalent to an RPQ $\regpath{R}^{-1}(X, Y)$.)
The answer of the \crpq $Q(X, Y) = \regpath{R}(Y, X)$ on an edge-labeled graph $G$ is the same as the answer of the \rpq $\regpath{R}^{-1}(X, Y)$ on the transpose graph $G^T$.
\end{restatable}

\begin{restatable}{proposition}{PropFilterRPQ}
\label{prop:filter-rpq}
(The query $Q(X, Y) = \regpath{R}(X, Y) \wedge S(Y)$ is equivalent to an \rpq.)
Let $\regpath{R}(X, Y)$ be an \rpq over an edge-labeled graph $G=(V, E, \Sigma)$,
and $S \subseteq V$ be a set of vertices viewed as a unary relation.
The answer of the query $Q(X, Y) = \regpath{R}(X, Y) \wedge S(Y)$ on $G$ is the same as the answer of an \rpq $\regpath{R}'(X, Y)$ on an edge-labeled graph $G'=(V, E', \Sigma')$
that can be constructed from $G$ in linear time.
\end{restatable}

\subsection{Width Measures}
\label{subsec:width-measures}
The following width measures were originally defined for conjunctive queries (CQs), but they extend naturally to \crpqs.

Given a \crpq (or CQ) $Q$,
a {\em tree decomposition} of $Q$ is
a pair $(T,\chi)$, where $T$ is a tree and $\chi: \nodes(T) \rightarrow 2^{\vars(Q)}$ is a
map from the nodes of $T$ to subsets of $\vars(Q)$ that satisfies the following properties:
\begin{itemize}[leftmargin=*]
    \item For every atom $\regpath{R}(X,Y)\in\atoms(Q)$, there is a node $t \in \nodes(T)$ such that $X,Y \in \chi(t)$. (In case of a CQ $Q$, for every atom $R(\bm X) \in \atoms(Q)$, there is a node $t \in \nodes(T)$ such that  $\bm X \subseteq \chi(t)$.)

    \item For every variable $X \in \vars(Q)$, the set $\setof{t}{X \in \chi(t)}$ forms a connected sub-tree of $T$.
\end{itemize}
Each set $\chi(t)$ is called a {\em bag} of the tree decomposition.
Let $\td(Q)$ denote the set of all tree decompositions of $Q$.
A tree decomposition $(T,\chi)$ of $Q$ is called {\em free-connex}~\cite{BaganDG07,10.1145/2448496.2448498} if
it contains a (possibly empty) connected subtree $T_f$ with nodes $\nodes(T_f)$ that contains all and only the free variables of $Q$, i.e., $\bigcup_{t \in \nodes(T_f)} \chi(t) = \free(Q)$.
Let $\ftd(Q)$ denote the set of all free-connex tree decompositions of $Q$.
Note that if the query is {\em Boolean} (i.e. $\free(Q)=\emptyset$)
or {\em full} (i.e. $\free(Q)=\vars(Q)$),
then every tree decomposition of $Q$ is free-connex,
in which case $\ftd(Q) = \td(Q)$.

Given a multigraph (or hypergraph) $\calG=(\calV, \calE)$, and a set of vertices $\bm Z \subseteq \calV$,
the {\em fractional edge cover number} of $\bm Z$, $\rho^*_{\calG}(\bm Z)$, is defined as the optimal value of the following linear program:
\begin{align}
    \text{minimize} \sum_{e \in \calE} x_e \label{eq:fractional-edge-cover}\quad\quad
    \text{subject to} \quad  \sum_{e: v \in e} x_e \geq 1, \quad \forall v \in \bm Z \quad\quad\text{and}\quad\quad
    x_e \geq 0, \quad \forall e \in \calE
\end{align}
Given a \crpq (or CQ) $Q$ and a set of variables $\bm Z \subseteq \vars(Q)$,
we use $\rho_Q^*(\bm Z)$ to denote $\rho^*_{\calG}(\bm Z)$ where $\calG$ is the query graph of $Q$.

\begin{definition}
[Free-Connex Fractional Hypertree Width]
Given a \crpq (or CQ) $Q$, the {\em fractional hypertree width}, and the {\em free-connex fractional hypertree width} of $Q$, are defined respectively as follows:
\begin{align}
    \fhtw(Q) &\quad\defeq\quad \min_{(T,\chi)\in\td(Q)}\quad \max_{t \in \nodes(T)}\quad \rho^*_Q(\chi(t)),\label{eq:fhtw}\\
    {\color{red}\fnfhtw}(Q) &\quad\defeq\quad \min_{(T,\chi)\in{\color{red}\ftd}(Q)}\quad \max_{t \in \nodes(T)}\quad \rho^*_Q(\chi(t)).\label{eq:fn-fhtw}
\end{align}
\label{defn:fn-fhtw}
\end{definition}

We state below some basic properties of the free-connex fractional hypertree width.
Similar to  $\rho^*_Q(\bm Z)$ (and $\rho^*_{\calG}(\bm Z)$) defined above,
we define the {\em integral} edge cover number $\rho_Q(\bm Z)$ (and $\rho_{\calG}(\bm Z)$)
by replacing the constraint $x_e \geq 0$ in the linear program~\eqref{eq:fractional-edge-cover} with $x_e \in \{0,1\}$, thus making it an integer linear program.
We define the {\em generalized hypertree width} $\ghtw(Q)$, and the {\em free-connex generalized hypertree width} $\fnghtw(Q)$, similar to Equations~\eqref{eq:fhtw} and~\eqref{eq:fn-fhtw} respectively, where we replace $\rho^*_Q$ with $\rho_Q$.
An acyclic \crpq (or acyclic \cq) is {\em free-connex} if it has a free-connex tree decomposition
where for every bag $\chi(t)$, there is an atom containing all variables in $\chi(t)$.
The following properties of the free-connex fractional hypertree width are folklore and have been proved in various forms in the literature, e.g.,~\cite{DBLP:journals/tods/OlteanuZ15,faq,xiao-os-yannakakis,BaganDG07}.

\begin{proposition}
    \label{prop:fn-fhtw-properties}
    For any acyclic \crpq (or acyclic CQ) $Q$, the following claims hold:
    \begin{itemize}
        \item $\fnfhtw(Q)$ is an integer. In particular, $\fnfhtw(Q) = \fnghtw(Q)$.
        \item $\fnfhtw(Q)$ is 1 if and only if $Q$ is free-connex.
    \end{itemize}
\end{proposition}

\begin{definition}[Trivial \crpqs]
    An acyclic \crpq $Q$ is called {\em trivial} if each connected subquery of $Q$ has at most one free variable.
    \label{defn:trivial-crpq}
\end{definition}
Trivial acyclic \crpqs are known to be solvable in time $O(N+\out)$~\cite{crpq-icdt26}.

\begin{restatable}[$\fnfhtw$ of $k$-expansion]{proposition}{PropKExpansionFnfhtw}
    Given an acyclic \crpq $Q$ and a natural number $k \geq 2$,
    let $Q_k'$ be the $k$-expansion of $Q$ (Definition~\ref{defn:k-expansion}).
    The following holds:
    \begin{align}
        \fnfhtw(Q_k') \quad=\quad \begin{cases}
            1 & \text{if } Q \text{ is trivial},\\
            \max(\fnfhtw(Q), 2) & \text{if } Q \text{ is non-trivial}.
        \end{cases}
    \end{align}
    \label{prop:k-expansion-fnfhtw}
\end{restatable}
As an example, the \crpq $Q(X, Y) = \regpath{R}(X, Y)$ has $\fnfhtw(Q) = 1$,
but its $k$-expansion is $Q_{k\text{-path}}$ (Eq.~\eqref{eq:intro-k-path}), which has
$\fnfhtw(Q_{k\text{-path}}) = 2$ for all $k \geq 2$.

\subsection{Data Statistics}
Given a set of variables $\bm X$, we refer to a tuple $\bm x$ over $\bm X$ as an {\em $\bm X$-tuple}.
\begin{definition}[Degrees in a relation]
    \label{defn:degree}
    Let $R(\bm Z)$ be a relation over a variable set $\bm Z$,
    and $\bm X, \bm Y \subseteq \bm Z$ be two (disjoint) subsets.
    Given an $\bm X$-tuple $\bm x$,
    the {\em degree of $\bm Y$ given $\bm X=\bm x$ in $R$}, denoted by $\deg_R(\bm Y|\bm X=\bm x)$, is the number of different $\bm Y$-tuples $\bm y$ that appear in $R$ together with $\bm x$. The {\em degree of $\bm Y$ given $\bm X$ in $R$}, denoted by $\deg_R(\bm Y|\bm X)$, is the maximum degree of $\bm Y$ given $\bm X = \bm x$ in $R$ over all $\bm X$-tuples $\bm x$:
    \begin{align*}
        \deg_R(\bm Y|\bm X=\bm x) \quad&\defeq\quad |\pi_{\bm Y}(\sigma_{\bm X=\bm x}(R))|,\\
        \deg_R(\bm Y|\bm X) \quad&\defeq\quad \max_{\bm x\in \pi_{\bm X}(R)} \deg_R(\bm Y|\bm X=\bm x).
    \end{align*}
\end{definition}

%% file: results.tex
\section{Results}
\label{sec:results}

The following theorem is our main result:
\begin{restatable}{theorem}{ThmMain}
    \label{thm:main}
    Let $Q$ be any acyclic {\crpq}, and $G$ be an input edge-labeled graph of size $N$.
    Then, $Q$ over $G$ can be answered in time $O(N + N\cdot \out^{1-\frac{1}{{\max(}\fnfhtw(Q){, 2)}}}+\out)$
    in data complexity, where $\out$ is the output size,
    and $\fnfhtw(Q)$ is the free-connex fractional hypertree width of $Q$ (Definition~\ref{defn:fn-fhtw}).
\end{restatable}

Our algorithm for acyclic \crpqs is {\em combinatorial}, i.e., it does not use fast matrix multiplication.
In order to assess how good its runtime is,
we compare it below to the best known runtime for acyclic \cqs over all combinatorial algorithms ~\cite{xiao-os-yannakakis,paris-os-yannakakis}.
\begin{theorem}[\cite{xiao-os-yannakakis}]
    \label{thm:acyclic-cq}
    Let $Q$ be any acyclic {\color{red}\cq} and $D$ be an input database instance of size $N$.
    Then, $Q$ over $D$ can be answered in time $O(N + N\cdot \out^{1-\frac{1}{\color{red}\fnfhtw(Q)}}+\out)$
    in data complexity, where $\out$ is the output size,
    and $\fnfhtw(Q)$ is the free-connex fractional hypertree width of $Q$.
\end{theorem}
There are two differences between the above two theorems (besides the obvious fact that one is for \crpqs while the other is for \cqs).
The first (subtle) difference
lies in the definition of acyclicity for \crpqs versus \cqs. See Remark~\ref{rem:acyclic_crpq} for a discussion.
The other difference is in the runtime expression, where in the \crpq case,
we take the maximum of $\fnfhtw(Q)$ and $2$.
Taking this maximum is unavoidable by Proposition~\ref{prop:k-expansion-fnfhtw}.
This is because solving a \crpq $Q$ is at least as hard as solving its $k$-expansion $Q_k'$ (which is a \cq) for every $k \geq 2$.
Excluding {\em trivial}\footnote{See Definition~\ref{defn:trivial-crpq}. Trivial \crpqs are known to be solvable in $O(N+\out)$ time~\cite{crpq-icdt26}.} \crpqs,
Proposition~\ref{prop:k-expansion-fnfhtw} says $\fnfhtw(Q'_k) = \max(\fnfhtw(Q), 2)$.
This shows that it is not possible to improve upon our runtime for acyclic \crpqs~ {\em without} improving the runtimes of the state-of-the-art combinatorial algorithms for acyclic \cqs.
In other words, our result links the fate of output-sensitive evaluation of acyclic \crpqs to that of acyclic \cqs.

We now prove our Theorem~\ref{thm:main}.
To that end, we will break it down into two steps:
In the first step, we reduce the problem of answering an acyclic \crpq to answering ``\freeleaf'' \crpqs (Lemma~\ref{lmm:decomposition}).
In the second step, we show how to answer \freeleaf \crpqs efficiently (Lemma~\ref{lmm:free-leaf}).

\begin{definition}[A \freeleaf \crpq]
    \label{defn:free-leaf}
    A \crpq $Q$ is called {\em \freeleaf} if it is acyclic, connected, and has a set of free variables which is exactly the leaves of its query graph.
\end{definition}

\begin{restatable}[Reduction from an acyclic \crpq to \freeleaf \crpqs]{lemma}{LmmDecomposition}
    \label{lmm:decomposition}
    Let $Q$ be any acyclic \crpq,
    and $G$ be an input edge-labeled graph of size $N$.
    Then, in $O(N)$-time in data complexity,
    we can construct $k$ \freeleaf \crpqs $Q_1, \ldots, Q_k$ and
    $k$ corresponding edge-labeled graphs $G_1, \ldots, G_k$
    for some constant $k$, satisfying the following:
    \begin{itemize}
        \item $|\free(Q_i)| \leq \max(\fnfhtw(Q), 2)$, for each $i \in [k]$.
        \item The output size of each $Q_i$ over $G_i$ is at most $\out$, where $\out$ is the output size of $Q$ over $G$.
        \item The answer of $Q$ over $G$ can be computed from the answers of $Q_1, \ldots, Q_k$ over $G_1, \ldots, G_k$ in time $O(\out)$.
    \end{itemize}
\end{restatable}
\begin{restatable}[Solving \freeleaf \crpqs]{lemma}{LmmFreeLeaf}
    \label{lmm:free-leaf}
    Let $Q$ be any \freeleaf \crpq, and $G$ be an input edge-labeled graph of size $N$.
    Then, $Q$ can be answered in time $O(N+N\cdot \out^{1-\frac{1}{|\free(Q)|}})$
    in data complexity, where 
    $\out$ is the output size.
\end{restatable}

\subsection{Proof of Lemma~\ref{lmm:decomposition}}
As a running example in this section, we will use the following acyclic \crpq $Q$,
which is depicted in Figure \ref{fig:decomposition} (top-left):
\begin{align}
Q(D,E,F,G,K,L) =& \regpath{R_1}(A,B) \wedge \regpath{R_2}(A, C) \wedge \regpath{R_3}(A,D) \wedge \regpath{R_4}(C, E) \wedge \regpath{R_5}(C,F)
\wedge \regpath{R_6}(D, G)\wedge\nonumber\\
& \regpath{R_7}(E, H) \wedge \regpath{R_8}(F, I) \wedge \regpath{R_9}(F, J)
\wedge \regpath{R_{10}}(H, K) \wedge \regpath{R_{11}}(H, L)
\label{eq:example-decomposition}
\end{align}
\input{fig-decomposition}
First, we introduce some concepts.
\begin{definition}[\Boundconnected \crpq]
    \label{defn:bound-connected}
    A \crpq (or \cq) $Q$ is called {\em \boundconnected} if every pair of variables in $\vars(Q)$ is connected in the query graph of $Q$
    by a path  that goes through only bound variables, $\bound(Q)$.
\end{definition}
Note that every \freeleaf \crpq is \boundconnected, but not every acyclic \boundconnected  \crpq is \freeleaf.
Nevertheless, every acyclic \boundconnected  \crpq can be made \freeleaf in linear time, as we will see later (Prop.~\ref{prop:bound-connected-to-free-leaf}).
\begin{definition}[\Boundconnected component of a \crpq]
    \label{defn:bound-connected-components}
    A {\em \boundconnected component} of a \crpq (or \cq) $Q$ is a {\em maximal} \boundconnected subquery $Q'$ of $Q$, i.e., a \boundconnected subquery $Q'$ such that any subquery $Q''$ of $Q$ with $\atoms(Q'') \supset \atoms(Q')$ is not \boundconnected.
\end{definition}
Figure \ref{fig:decomposition} (top-right) depicts the \boundconnected components of $Q$ from Eq.~\eqref{eq:example-decomposition}, which are:
\begin{align}
Q_1(D, E, F) &= \regpath{R_1}(A,B) \wedge \regpath{R_2}(A, C) \wedge \regpath{R_3}(A, D) \wedge \regpath{R_4}(C, E) \wedge \regpath{R_5}(C, F) \label{ex:decomposing:component1}\\
Q_2(D, G) &= \regpath{R_6}(D,G) \label{ex:decomposing:component2}\\
Q_3(E, K, L) &= \regpath{R_7}(E,H) \wedge \regpath{R_{10}}(H, K) \wedge \regpath{R_{11}}(H, L) \label{ex:decomposing:component3}\\
Q_4(F) &= \regpath{R_{8}}(F, I) \wedge \regpath{R_{9}}(F, J)
\label{ex:decomposing:component4}
\end{align}

The following statement was proven for conjunctive queries in prior work  (Lemmas 3.3 and 3.4 in~\cite{xiao-os-yannakakis}). 
We give a simplified proof adapted to \crpqs and using variable elimination orders~\cite{faq,DBLP:journals/tods/KhamisCMNNOS20}.

\begin{restatable}[\cite{xiao-os-yannakakis}]{lemma}{LemFhtwConComp}
\label{lem:fhtw_con_comp}   
    Let $Q$ be an acyclic \crpq (or acyclic \cq) and $Q_1, \ldots, Q_k$ its \boundconnected components.\footnote{The \boundconnected components of an acyclic \crpq are also acyclic.}
    Then, we have
    \begin{align}
        \fnfhtw(Q) = \max(\max_{i \in [k]}\rho^*_{Q_i}(\free(Q_i)), 1)
        \label{eq:fhtw-con-comp}
    \end{align}
\end{restatable}

\begin{restatable}{proposition}{PropNumberFreeRhoStar}
    \label{prop:number_free_rho_star}
    For any acyclic \boundconnected \crpq $Q$,
    $|\free(Q)|\leq\max(\rho^*_Q(\free(Q)), 2)$.
\end{restatable}
Lemma~\ref{lem:fhtw_con_comp} and Prop.~\ref{prop:number_free_rho_star} imply that for any
\boundconnected component $Q_i$ of an acyclic \crpq $Q$,
we must have
\begin{align}
    |\free(Q_i)| \quad\leq\quad \max(\fnfhtw(Q), 2)
    \label{eq:free-to-fnfhtw}
\end{align}
Consider for example $Q$ depicted in Figure \ref{fig:decomposition}.
The free-connex fractional hypertree width of $Q$ is $\fnfhtw(Q) = 3$,
whereas the components $Q_1, Q_2, Q_3, Q_4$ from Eqs.~\eqref{ex:decomposing:component1}--\eqref{ex:decomposing:component4}
have $3, 2, 3, 1$ free variables, respectively.

\begin{restatable}{proposition}{PropBoundConnectedToLeafFree}
    \label{prop:bound-connected-to-free-leaf}
    (Every acyclic \boundconnected \crpq can be transformed into \freeleaf.)
    For any acyclic \boundconnected  CRPQ $Q$ and input edge-labeled graph $G = (V, E, \Sigma)$ of size $N$, we can construct a \freeleaf CRPQ $Q'$ and an edge-labeled graph $G' = (V, E', \Sigma')$ in time
    $O(N)$ in data complexity such that $\free(Q) = \free(Q')$ and the answer of $Q$ over $G$ is the same as the answer of $Q'$ over $G'$.
\end{restatable}
The above proposition can be proved by repeatedly removing non-free leaf variables from $Q$ using Propositions~\ref{prop:linear-rpq} and~\ref{prop:filter-rpq}.
Figure \ref{fig:decomposition} (bottom-right) depicts the \freeleaf \crpqs obtained by applying Prop.~\ref{prop:bound-connected-to-free-leaf} to each \boundconnected component of $Q$ from Eq.~\eqref{eq:example-decomposition}.

Our reduction in Lemma~\ref{lmm:decomposition}
relies on taking the \boundconnected components of $Q$,
and applying Prop.~\ref{prop:bound-connected-to-free-leaf} to each component.
But there is a little twist:
Consider for example the component $Q_1$ from Eq.~\eqref{ex:decomposing:component1}.
The output size of $Q_1$ could potentially be larger than the output size of $Q$,
thus violating the second bullet point in Lemma~\ref{lmm:decomposition}.
This is because, for example, the regular path query $\regpath{R_6}(D, G)$
could produce an empty output, but $\regpath{R_6}(D, G)$ is not part of $Q_1$.
To resolve this issue, we have to modify $Q_1$ a bit. In particular,
we define the following \crpq, $Q_D$, whose body is the same as $Q$ but has only one free variable, $D$:
\begin{align*}
Q_D(D) =& \regpath{R_1}(A,B) \wedge \regpath{R_2}(A, C) \wedge \regpath{R_3}(A,D) \wedge \regpath{R_4}(C, E) \wedge \regpath{R_5}(C,F)
\wedge \regpath{R_6}(D, G)\wedge\nonumber\\
& \regpath{R_7}(E, H) \wedge \regpath{R_8}(F, I) \wedge \regpath{R_9}(F, J)
\wedge \regpath{R_{10}}(H, K) \wedge \regpath{R_{11}}(H, L)
\label{eq:example-decomposition:QD}
\end{align*}
Similarly, we define a \crpq $Q_X$ for every free variable $X$ of $Q$.
Each such \crpq $Q_X$ can be answered in linear time by repeatedly applying
Propositions~\ref{prop:linear-rpq} and~\ref{prop:filter-rpq} to eliminate leaf variables
that are not $X$, as the following proposition shows.
\begin{proposition}[\cite{crpq-icdt26}]
\label{prop:linear-time-crpq}
    Every acyclic \crpq $Q$ with at most one free variable can be answered in time $O(N)$
    in data complexity, where $N$ is the size of the input edge-labeled graph.
\end{proposition}
After evaluating each $Q_X$,
we extend $Q_1$ by adding the outputs of the \crpqs $Q_D, Q_E$, and $Q_F$ as extra filters to the body of $Q_1$:
\begin{align*}
Q_1'(D, E, F) &= \regpath{R_1}(A,B) \wedge \regpath{R_2}(A, C) \wedge \regpath{R_3}(A, D) \wedge \regpath{R_4}(C, E) \wedge \regpath{R_5}(C, F) {\color{red}\wedge Q_D(D) \wedge Q_E(E) \wedge Q_F(F)}
\end{align*}
Now, the output of $Q_1'$ above is guaranteed to be the projection of the output of $Q$ on the variables $(D, E, F)$.
By Proposition~\ref{prop:filter-rpq}, $Q_1'$ above is still equivalent to a \crpq:
\begin{align*}
Q_1'(D, E, F) &= \regpath{R_1}(A,B) \wedge \regpath{R_2}(A, C) \wedge \regpath{R_3'}(A, D) \wedge \regpath{R_4'}(C, E) \wedge \regpath{R_5'}(C, F)
\end{align*}
The above \crpq is acyclic \boundconnected but not yet \freeleaf because $B$ is a leaf variable and it is not free.
By applying Proposition~\ref{prop:linear-rpq}, followed by Proposition~\ref{prop:filter-rpq},
we can eliminate $B$, thus making $Q_1'$ \freeleaf:
\begin{align*}
Q_1'(D, E, F) &= \regpath{R_2'}(A, C) \wedge \regpath{R_3'}(A, D) \wedge \regpath{R_4'}(C, E) \wedge \regpath{R_5'}(C, F)
\end{align*}
Because it is \freeleaf, the above query can be solved using the algorithm from Lemma~\ref{lmm:free-leaf} in time $O(N + N \cdot \out_1^{1-\frac{1}{|\free(Q_1')|}})$, where $\out_1$
is the output size of $Q_1'$, which satisfies $\out_1 \leq \out$ by construction.
By Eq.~\eqref{eq:free-to-fnfhtw}, this is bounded by $O(N + N \cdot \out^{1-\frac{1}{\max(\fnfhtw(Q), 2)}})$,
as desired.

Similarly, we can define \freeleaf \crpqs $Q_2'(D, G), Q_3'(E, K, L), Q_4'(F)$
and solve them in similar complexity.
Finally, the answers from all components $Q_i'$ form an acyclic full conjunctive query,
which is depicted in Figure~\ref{fig:decomposition} (bottom-left).
We can solve it using Yannakakis algorithm~\cite{Yannakakis81} in time $O(\sum_i \out_i + \out) = O(\out)$.


\subsection{Proof of Lemma~\ref{lmm:free-leaf}}
\label{subsec:proof-of-lmm-free-leaf}
In this section, we demonstrate our proof of Lemma~\ref{lmm:free-leaf}.
First, we introduce a definition and some auxiliary propositions.

\begin{definition}[$(X,\Delta)$-restriction]
    Let $R(X, \bm Y)$ be a relation over variables $\{X\} \cup \bm Y$, where $X$ is a single variable
    and $\bm Y$ is the set of remaining variables of $R$.
    Let $\Delta$ be a natural number.
    A relation $\ov R(X, \bm Y)$ is called an {\em $(X,\Delta)$-restriction} of $R(X, \bm Y)$
    if it is a subset of $R(X, \bm Y)$ that satisfies the following for every $X$-value $x$ that appears in $R$:
    \begin{itemize}
        \item If $\deg_{R}(\bm Y|X=x) \leq \Delta$, then $\ov R$ contains all tuples of $R$ with $X$-value $x$.
        \item If $\deg_{R}(\bm Y|X=x) > \Delta$, then $\ov R$ contains exactly $\Delta$ tuples of $R$ with $X$-value $x$.
    \end{itemize}
\end{definition}
An $(X,\Delta)$-restriction of $R$ is not necessarily unique, as there can be several subsets of $\Delta$ tuples of $R$ with $X$-value $x$.

\begin{restatable}[Propagation of restrictions through \rpqs]{proposition}{PropRestrictionPropagation}
    \label{prop:restriction-propagation}
    Let $G=(V, E,\Sigma)$ be an input edge-labeled graph of size $N$
    and $\Delta$ a natural number.
    Consider the query 
    \begin{align*}
        Q(X, \bm Z) = \regpath{R}(X, Y) \wedge S(Y, \bm Z),
    \end{align*}
    where $\regpath{R}(X, Y)$ is an \rpq over the alphabet $\Sigma$,
    and $S(Y, \bm Z)$ is a relation over variables $\{Y\} \cup \bm Z$.
    Then, given a $(Y,\Delta)$-restriction $\ov S(Y, \bm Z)$ of $S(Y, \bm Z)$,
    we can compute an $(X, \Delta)$-restriction $\ov Q(X, \bm Z)$ of $Q(X, \bm Z)$ in time $O(N \cdot \Delta)$.~\footnote{The size of $\ov S$ could be much larger than $N$, yet, the runtime still holds. However, we also need to assume that $\ov S(Y, \bm Z)$ is indexed by a {\em dictionary} that returns for a given $X$-value, the corresponding list of up to $\Delta$ $\bm Z$-tuples.}

\end{restatable}
Prop.~\ref{prop:restriction-propagation} implies the following as a special case.
\begin{proposition}[$(X,\Delta)$-restriction of a single \rpq]
    \label{prop:restriction-propagation:basecase}
    For any \rpq $\regpath{R}(X, Y)$ and any natural number $\Delta$,
    we can compute an $(X, \Delta)$-restriction $\bar{R}(X,Y)$ of $\regpath{R}(X, Y)$
    in time $O(N \cdot \Delta)$.
\end{proposition}

Prop.~\ref{prop:restriction-propagation:basecase}
follows from Prop.~\ref{prop:restriction-propagation}
because we can always construct an identity relation $S(Y, Z)\defeq \{(v, v) \mid v \in V\}$ and apply Prop.~\ref{prop:restriction-propagation} on the query $Q(X, Z) = \regpath{R}(X, Y) \wedge S(Y, Z)$.


\begin{proposition}[Composition of $(X,\Delta)$-restrictions]
    \label{prop:restriction-composition}
    Let $R_1(X, \bm Z_1)$ and $R_2(X, \bm Z_2)$ be two relations over variables $\{X\} \cup \bm Z_1$ and $\{X\} \cup \bm Z_2$, respectively, where $\bm Z_1$ and $\bm Z_2$ are {\em disjoint}.
    Consider the following query:
    \begin{align}
        Q(X, \bm Z_1, \bm Z_2) = R_1(X, \bm Z_1) \wedge R_2(X, \bm Z_2)
    \end{align}
    Given $(X, \Delta)$-restrictions $\ov R_1(X, \bm Z_1)$ 
    of $R_1(X, \bm Z_1)$ and  $\ov R_2(X, \bm Z_2)$ of $R_2(X, \bm Z_2)$,
    where $\Delta$ is some natural number,
    an $(X,\Delta)$-restriction $\ov Q(X, \bm Z_1, \bm Z_2)$ of $Q(X, \bm Z_1, \bm Z_2)$ can be computed in time $O(|\dom_X| \cdot \Delta)$,
    where $\dom_X$ is the set of all $X$-values appearing in $R_1$ and $R_2$.
\end{proposition}

We next demonstrate
the idea of the proof of Lemma~\ref{lmm:free-leaf} using the following example.
We state the general proof in Appendix~\ref{app:results:free-leaf}.
\begin{example}[for Lemma~\ref{lmm:free-leaf}]
    \label{ex:free-leaf}
    Consider the following \freeleaf \crpq $Q$ depicted in Fig.~\ref{subfig:ex-free-leaf:original}:
    \begin{align}
        Q(A, B, C, D) = \regpath{R_1}(X, Y) \wedge \regpath{R_2}(X, Z) \wedge \regpath{R_3}(Y, A) \wedge \regpath{R_4}(Y, B) \wedge \regpath{R_5}(Z, C) \wedge \regpath{R_6}(Z, D)
        \label{eq:ex-free-leaf}
    \end{align}
    \input{fig-free-leaf}

    Pick an arbitrary free variable, say $D$, and re-orient the query tree so that $D$ becomes the root, as shown in Fig.~\ref{subfig:ex-free-leaf:rotated}.
    Note that when we re-orient an \rpq, e.g., $\regpath{R_6}(Z, D)$, so that it is from $D$ to $Z$,
    the result is still an \rpq (Prop.~\ref{prop:rpq-transpose}).
    To reduce clutter, we  abuse notation and refer to the re-oriented \rpq $\regpath{R_6}(Z, D)$ as $\regpath{R_6}(D, Z)$.

    Using a bottom-up traversal of the query tree from the leaves to the root $D$,
    we can rewrite the query $Q$ as a sequence of subqueries as follows:
    \begin{align*}
        P_1(Y, A) &\quad\defeq\quad \regpath{R_3}(Y, A) &
            \text{(Prop.~\ref{prop:restriction-propagation:basecase})}\\
        P_2(Y, B) &\quad\defeq\quad \regpath{R_4}(Y, B) &
            \text{(Prop.~\ref{prop:restriction-propagation:basecase})}\\
        P_3(Y, A, B) &\quad\defeq\quad P_1(Y, A) \wedge P_2(Y, B) &
            \text{(Prop.~\ref{prop:restriction-composition})}\\
        P_4(X, A, B) &\quad\defeq\quad \regpath{R_1}(X, Y) \wedge P_3(Y, A, B) &
            \text{(Prop.~\ref{prop:restriction-propagation})}\\
        P_5(Z, A, B) &\quad\defeq\quad \regpath{R_2}(Z, X) \wedge P_4(X, A, B) &
            \text{(Prop.~\ref{prop:restriction-propagation})}\\
        P_6(Z, C) &\quad\defeq\quad \regpath{R_5}(Z, C) &
            \text{(Prop.~\ref{prop:restriction-propagation:basecase})}\\
        P_7(Z, A, B, C) &\quad\defeq\quad P_5(Z, A, B) \wedge P_6(Z, C) &
            \text{(Prop.~\ref{prop:restriction-composition})}\\
        Q(D, A, B, C) &\quad\defeq\quad \regpath{R_6}(D, Z) \wedge P_7(Z, A, B, C) &
            \text{(Prop.~\ref{prop:restriction-propagation})}
    \end{align*}

    We assume for now that we already know the output size $\out$ of $Q$ upfront.
    We will show at the end of the example how to get rid of this assumption.
    Let $\Delta\defeq \out^{3/4}$.
    (In general, $\Delta \defeq \out^{1-1/|\free(Q)|}$.)
    Moreover, let $\Delta' \defeq \Delta + 1$.
    For each query $P_i(X, \bm Y)$, where $X$ is the first variable and $\bm Y$ is the set of remaining variables,
    we compute an $(X, \Delta')$-restriction $\ov P_i(X, \bm Y)$ of $P_i(X, \bm Y)$
    using one of Props.~\ref{prop:restriction-propagation},~\ref{prop:restriction-propagation:basecase},
    or~\ref{prop:restriction-composition}, in time $O(N \cdot \Delta')$.
    At the very end, we compute a $(D, \Delta')$-restriction $\ov Q(D, A, B, C)$ of $Q(D, A, B, C)$.
    A vertex $d$ is called {\em light} if $\deg_{\ov Q}(A,B,C|D=d) \leq \Delta$,
    and {\em heavy} otherwise (i.e., if $\deg_{\ov Q}(A,B,C|D=d) = \Delta'\defeq \Delta+1$; \change{recall that $\deg_{\ov Q}(A,B,C|D=d)$ cannot exceed $\Delta'$ by definition of $(D, \Delta')$-restriction}).
    The relation $\ov Q$ already contains all output tuples $(d, a, b, c)$ where $d$ is light.
    Let $H_D$ be the set of all heavy vertices $d$.
    Our next target is to report output tuples $(d, a, b, c)$ where $d \in H_D$.
    Note that because each heavy vertex $d$ in $H_D$ has at least $\Delta$ associated tuples $(a, b, c)$ where $(d, a, b, c)$ is in the output,
    the number of heavy vertices, $|H_D|$, is at most $\out/\Delta = \out^{1/4}$.
    (In general, we get $|H_D| \leq \out^{1/|\free(Q)|}$.)
    Using Prop.~\ref{prop:filter-rpq}, we can assume that the query $\regpath{R_6}(D, Z) \wedge H_D(D)$
    is equivalent to an \rpq $\regpath{R_6'}(D, Z)$.
    Before we continue, we replace $\regpath{R_6}(D, Z)$ with $\regpath{R_6'}(D, Z)$,
    and from now on,  we can assume that in the remaining query, the number of different $D$-values is at most $\out^{1/4}$.

    Now, we pick another free variable, say $C$, and repeat the same process.
    Just like before, at the end of the bottom-up traversal, we will be able to report all output tuples $(c, a, b, d)$
    where $c$ is light, and we will be left with a set of heavy vertices $H_C$ of size at most
    $\out^{1/4}$.
    Our next target is to report output tuples $(c, a, b, d)$ where $c \in H_C$ (and $d \in H_D$).
    In the remaining query, we can assume that the number of different $C$-values (and $D$-values) is at most $\out^{1/4}$.

    Now, we pick another free variable, say $B$, and repeat the same process. After we report all output tuples $(b, a, c, d)$ where $b$ is light, we can assume that in the remaining query,
    the number of different $B$-values (as well as $C$-values and $D$-values) is at most $\out^{1/4}$.

    Finally, we pick the last free variable, $A$. But now, since the number of different $B$-, $C$-, and $D$-values is at most $\out^{1/4}$, the total number of different tuples $(b, c, d)$ cannot exceed
    $\Delta = \out^{3/4}$. Hence, in all remaining output tuples $(a, b, c, d)$,
    the $A$-value $a$ must be light, and we can report all of them at the end of the bottom-up traversal.

    \paragraph*{Guessing the output size $\out$.}
    We are left to explain how to guess the output size $\out$ if it is not given upfront.
    To that end, we start with an initial guess $\out = 1$ and keep doubling our guess
    every time we discover that our guess was below the actual output size.
    In particular, for each guess of $\out$, we run the above algorithm.
    Consider the last bottom-up traversal above where the free variable $A$ is the root.
    If it is indeed the case that all remaining $A$-values are light,
    then we already reported the whole output of $Q$ and we are done.
    If not, then our guess of $\out$ was below the actual output size,
    and we need to double it and repeat the process.
    The total time spent is $O(\sum_{i=1, \ldots, \lceil\log \out\rceil} N \cdot 2^{3i/4})$.
    This is a geometric series that is dominated by its last term,
    which is $O(N \cdot \out^{3/4})$.
\end{example}




%% file: fig-decomposition.tex
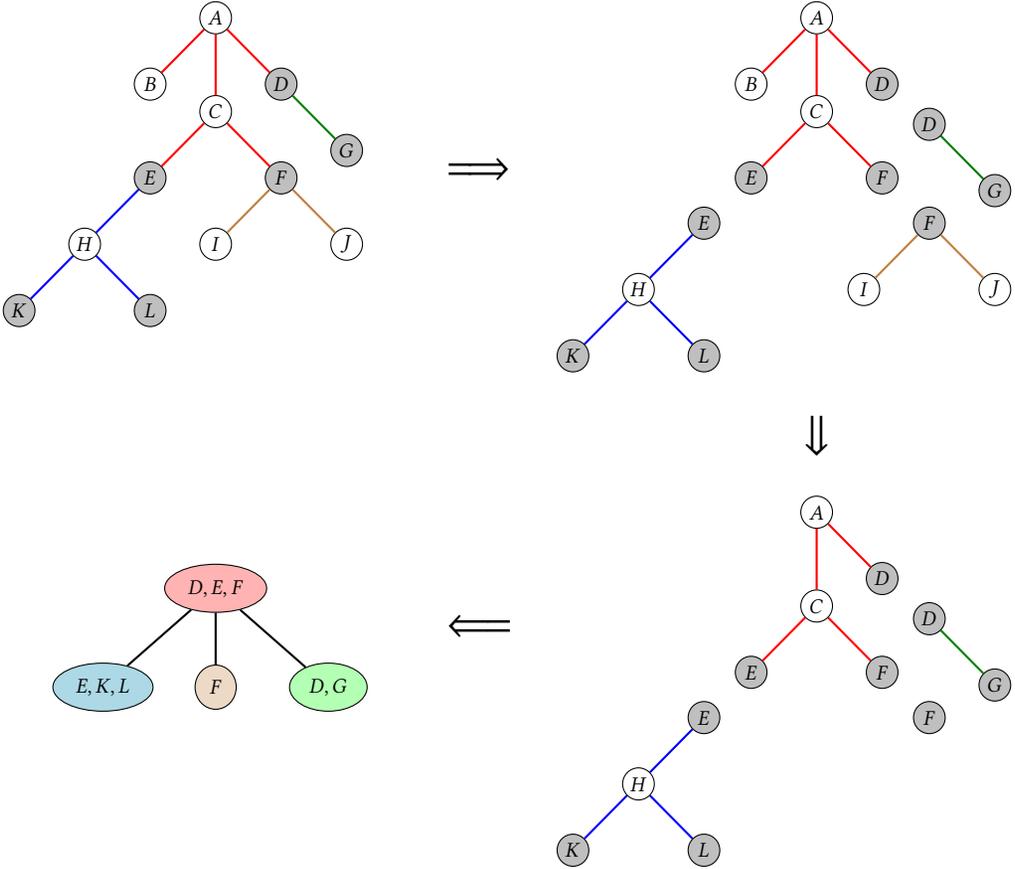
\begin{figure}[th]
    \centering
    \begin{tikzpicture}[every node/.style={minimum size = 15pt, inner sep=1pt, scale =.8}]
        \node [circle, draw] at (0,0) (A) {$A$};
        \node [circle, draw, node distance=.8cm, below left=of A] (B) {$B$};
        \node [circle, draw, node distance=.8cm, below =of A] (C) {$C$};
        \node [circle, draw, fill=lightgray, node distance=.8cm, below right=of A] (D) {$D$};
        \node [circle, draw, fill=lightgray, node distance=.8cm, below left=of C] (E) {$E$};
        \node [circle, draw, fill=lightgray, node distance=.8cm, below right=of C] (F) {$F$};
        \node [circle, draw, fill=lightgray, node distance=.8cm, below right=of D] (G) {$G$};
        \node [circle, draw, node distance=.8cm, below left=of E] (H) {$H$};
        \node [circle, draw, node distance=.8cm, below left=of F] (I) {$I$};
        \node [circle, draw, node distance=.8cm, below right=of F] (J) {$J$};
        \node [circle, draw, fill=lightgray, node distance=.8cm, below left=of H] (K) {$K$};
        \node [circle, draw, fill=lightgray, node distance=.8cm, below right=of H] (L) {$L$};
        \draw[thick, red] (A) -- (B);
        \draw[thick, red] (A) -- (C);
        \draw[thick, red] (A) -- (D);
        \draw[thick, red] (C) -- (E);
        \draw[thick, red] (C) -- (F);

        \draw[thick, DarkGreen] (D) -- (G);
        
        \draw[thick, blue] (E) -- (H);
        \draw[thick, blue] (H) -- (K);
        \draw[thick, blue] (H) -- (L);

        \draw[thick, brown] (F) -- (I);
        \draw[thick, brown] (F) -- (J);

        \node at (3.5,-2) {\Huge$\Longrightarrow$};

        \begin{scope}[shift={(1,0)}]
            \begin{scope}[shift={(7,0)}]
                \node [circle, draw] at (0,0) (A) {$A$};
                \node [circle, draw, node distance=.8cm, below left=of A] (B) {$B$};
                \node [circle, draw, node distance=.8cm, below =of A] (C) {$C$};
                \node [circle, draw, fill=lightgray, node distance=.8cm, below right=of A] (D) {$D$};
                \node [circle, draw, fill=lightgray, node distance=.8cm, below left=of C] (E) {$E$};
                \node [circle, draw, fill=lightgray, node distance=.8cm, below right=of C] (F) {$F$};

                \draw[thick, red] (A) -- (B);
                \draw[thick, red] (A) -- (C);
                \draw[thick, red] (A) -- (D);
                \draw[thick, red] (C) -- (E);
                \draw[thick, red] (C) -- (F);
            \end{scope}

            \begin{scope}[shift={(8.5,-1.4)}]
                \node [circle, draw, fill=lightgray] at (0,0) (D) {$D$};
                \node [circle, draw, fill=lightgray, node distance=.8cm, below right=of D] (G) {$G$};
                \draw[thick, DarkGreen] (D) -- (G);
            \end{scope}

            \begin{scope}[shift={(5.5,-2.7)}]
                \node [circle, draw, fill=lightgray] at (0,0) (E) {$E$};
                \node [circle, draw, node distance=.8cm, below left=of E] (H) {$H$};
                \node [circle, draw, fill=lightgray, node distance=.8cm, below right=of H] (L) {$L$};
                \node [circle, draw, fill=lightgray, node distance=.8cm, below left=of H] (K) {$K$};
                \draw[thick, blue] (E) -- (H);
                \draw[thick, blue] (H) -- (K);
                \draw[thick, blue] (H) -- (L);
            \end{scope}

            \begin{scope}[shift={(8.5,-2.7)}]
                \node [circle, draw, fill=lightgray]  at (0,0) (F) {$F$};
                \node [circle, draw, node distance=.8cm, below left=of F] (I) {$I$};
                \node [circle, draw, node distance=.8cm, below right=of F] (J) {$J$};
                \draw[thick, brown] (F) -- (I);
                \draw[thick, brown] (F) -- (J);
            \end{scope}
        \end{scope}

        \node at (8,-5.5) {\Huge$\Downarrow$};

        \begin{scope}[shift={(-1,-6.5)}]
            \begin{scope}[shift={(9,0)}]
                \node [circle, draw] at (0,0) (A) {$A$};
                \node [circle, draw, node distance=.8cm, below =of A] (C) {$C$};
                \node [circle, draw, fill=lightgray, node distance=.8cm, below right=of A] (D) {$D$};
                \node [circle, draw, fill=lightgray, node distance=.8cm, below left=of C] (E) {$E$};
                \node [circle, draw, fill=lightgray, node distance=.8cm, below right=of C] (F) {$F$};

                \draw[thick, red] (A) -- (C);
                \draw[thick, red] (A) -- (D);
                \draw[thick, red] (C) -- (E);
                \draw[thick, red] (C) -- (F);
            \end{scope}

            \begin{scope}[shift={(10.5,-1.4)}]
                \node [circle, draw, fill=lightgray] at (0,0) (D) {$D$};
                \node [circle, draw, fill=lightgray, node distance=.8cm, below right=of D] (G) {$G$};
                \draw[thick, DarkGreen] (D) -- (G);
            \end{scope}

            \begin{scope}[shift={(7.5,-2.7)}]
                \node [circle, draw, fill=lightgray] at (0,0) (E) {$E$};
                \node [circle, draw, node distance=.8cm, below left=of E] (H) {$H$};
                \node [circle, draw, fill=lightgray, node distance=.8cm, below right=of H] (L) {$L$};
                \node [circle, draw, fill=lightgray, node distance=.8cm, below left=of H] (K) {$K$};

                \draw[thick, blue] (E) -- (H);
                \draw[thick, blue] (H) -- (K);
                \draw[thick, blue] (H) -- (L);
            \end{scope}

            \begin{scope}[shift={(10.5,-2.7)}]
                \node [circle, draw, fill=lightgray]  at (0,0) (F) {$F$};
            \end{scope}
        \end{scope}

        \node at (3.5,-8) {\Huge$\Longleftarrow$};

        \begin{scope}[shift={(0, -7.5)}]
            \node [ellipse, draw, inner sep=4pt, fill=lightred] at (0,0) (DEF) {$D, E, F$};
            \node [ellipse, draw, inner sep=4pt, fill=lightgreen] at (1.5,-1.3) (DG) {$D, G$};
            \node [ellipse, draw, inner sep=4pt, fill=lightbrown] at (0,-1.3) (FF) {$F$};
            \node [ellipse, draw, inner sep=4pt, fill=lightblue] at (-1.5,-1.3) (EKL) {$E, K, L$};
            \draw[thick] (DEF) -- (FF);
            \draw[thick] (DEF) -- (DG);
            \draw[thick] (DEF) -- (EKL);
        \end{scope}

    \end{tikzpicture}
    \caption{An example demonstrating the reduction in Lemma~\ref{lmm:decomposition}.
    Top-left: The input acyclic \crpq $Q$ from Eq.~\eqref{eq:example-decomposition}, where
    free variables are \fcolorbox{black}{lightgray}{shaded}.
    Top-right: The bound connected components of $Q$ (Definition~\ref{defn:bound-connected}).
    Bottom-right: The corresponding \freeleaf \crpqs obtained by applying Proposition~\ref{prop:bound-connected-to-free-leaf} to each component.
    Each \freeleaf \crpq is solved using the algorithm from Lemma~\ref{lmm:free-leaf},
    and the answers are stitched together into an acyclic full conjunctive query over only the free variables of $Q$.
    This query is depicted at the bottom-left, and it can be solved using Yannakakis algorithm.}
    \Description{An example demonstrating the reduction in Lemma~\ref{lmm:decomposition}.
    Top-left: The input acyclic \crpq $Q$ from Eq.~\eqref{eq:example-decomposition}, where
    free variables are \fcolorbox{black}{lightgray}{shaded}.
    Top-right: The bound connected components of $Q$ (Definition~\ref{defn:bound-connected}).
    Bottom-right: The corresponding \freeleaf \crpqs obtained by applying Proposition~\ref{prop:bound-connected-to-free-leaf} to each component.
    Each \freeleaf \crpq is solved using the algorithm from Lemma~\ref{lmm:free-leaf},
    and the answers are stitched together into an acyclic full conjunctive query over only the free variables of $Q$.
    This query is depicted at the bottom-left, and it can be solved using Yannakakis algorithm.}
    \label{fig:decomposition}
\end{figure}

%% file: fig-free-leaf.tex
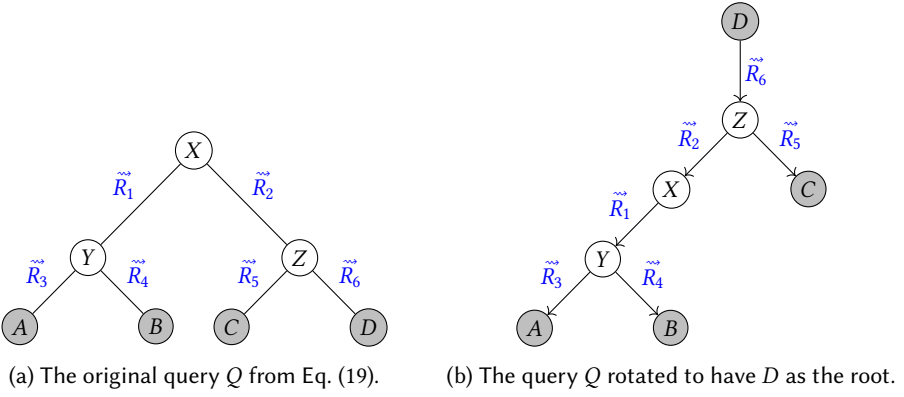
\begin{figure}[th]
    \begin{subfigure}[b]{0.45\textwidth}
        \centering
        \begin{tikzpicture}[every node/.style={inner sep=2pt, scale =.9}]
            \node [circle, draw] at (0,0) (X) {$X$};
            \node [circle, draw, node distance=1.5cm, below left=of X] (Y) {$Y$};
            \node [circle, draw, node distance=1.5cm, below right=of X] (Z) {$Z$};
            \node [circle, draw, fill = lightgray, node distance=.8cm, below left=of Y] (A) {$A$};
            \node [circle, draw, fill = lightgray, node distance=.8cm, below right=of Y] (B) {$B$};
            \node [circle, draw, fill = lightgray, node distance=.8cm, below left=of Z] (C) {$C$};
            \node [circle, draw, fill = lightgray, node distance=.8cm, below right=of Z] (D) {$D$};
            \draw[-] (X) -- node[above left] {\color{blue}$\regpath{R_1}$} (Y);
            \draw[-] (X) -- node[above right] {\color{blue}$\regpath{R_2}$} (Z);
            \draw[-] (Y) -- node[above left] {\color{blue}$\regpath{R_3}$} (A);
            \draw[-] (Y) -- node[above right] {\color{blue}$\regpath{R_4}$} (B);
            \draw[-] (Z) -- node[above left] {\color{blue}$\regpath{R_5}$} (C);
            \draw[-] (Z) -- node[above right] {\color{blue}$\regpath{R_6}$} (D);
        \end{tikzpicture}
        \caption{The original query $Q$ from Eq.~\eqref{eq:ex-free-leaf}.}
        \label{subfig:ex-free-leaf:original}
    \end{subfigure}
    \begin{subfigure}[b]{0.45\textwidth}
        \centering
        \begin{tikzpicture}[every node/.style={inner sep=2pt, scale =.9}]
            \node [circle, draw] at (0,0) (X) {$X$};
            \node [circle, draw, node distance=.8cm, below left=of X] (Y) {$Y$};
            \node [circle, draw, node distance=.8cm, above right=of X] (Z) {$Z$};
            \node [circle, draw, fill = lightgray, node distance=.8cm, below left=of Y] (A) {$A$};
            \node [circle, draw, fill = lightgray, node distance=.8cm, below right=of Y] (B) {$B$};
            \node [circle, draw, fill = lightgray, node distance=.8cm, below right=of Z] (C) {$C$};
            \node [circle, draw, fill = lightgray, node distance=.8cm, above=of Z] (D) {$D$};
            \draw[->] (X) -- node[above left] {\color{blue}$\regpath{R_1}$} (Y);
            \draw[->] (Z) -- node[above left] {\color{blue}$\regpath{R_2}$} (X);
            \draw[->] (Y) -- node[above left] {\color{blue}$\regpath{R_3}$} (A);
            \draw[->] (Y) -- node[above right] {\color{blue}$\regpath{R_4}$} (B);
            \draw[->] (Z) -- node[above right] {\color{blue}$\regpath{R_5}$} (C);
            \draw[->] (D) -- node[right] {\color{blue}$\regpath{R_6}$} (Z);
        \end{tikzpicture}
        \caption{The query $Q$ rotated to have $D$ as the root.}
        \Description{The query $Q$ rotated to have $D$ as the root.}
        \label{subfig:ex-free-leaf:rotated}
    \end{subfigure}
    \caption{The \freeleaf acyclic \crpq $Q$ from Example~\ref{ex:free-leaf}. Free variables are \fcolorbox{black}{lightgray}{shaded}.}
    \Description{The \freeleaf acyclic \crpq $Q$ from Example~\ref{ex:free-leaf}. Free variables are \fcolorbox{black}{lightgray}{shaded}.}
    \label{fig:ex-free-leaf}
\end{figure}

%% file: related.tex
\section{Related Work}
\label{sec:related}

The traditional approach for evaluating \rpqs is based on taking the product graph (Definition~\ref{defn:product-graph}) of the input graph with the NFA corresponding to the regular expression and then running a breadth-first search starting from each vertex in the product graph.
In data complexity, this approach takes time $O(|V| \cdot |E|)$ ~\cite{Baeza13, MartensT18, MendelzonW95}, which is $O(N^2)$, where $N\defeq |V| + |E|$.
Moreover, no combinatorial algorithm can achieve any polynomial improvement\footnote{By that, we mean achieving a runtime of $O((|V|\cdot |E|)^{1-\epsilon})$ for any $\epsilon > 0$.} over this bound unless
the {\em combinatorial Boolean matrix multiplication conjecture} fails~\cite{CaselS23}.
A prior work introduced an output sensitive algorithm for \rpqs whose complexity is $O(N + N \cdot \out^{1/2})$, where $\out$ is the output size~\cite{rpq-pods25}.

For \crpqs, prior work~\cite{DBLP:conf/icdt/CucumidesRV23} proposed solving them by first materializing each \rpq separately,
thus turning the \crpq into a traditional \cq, and then evaluating this \cq using
standard approaches.
This approach is not output-sensitive, since the output of each individual \rpq could be much larger than the final output of the \crpq.
An output-sensitive alternative for {\em acyclic} \crpqs was proposed in~\cite{crpq-icdt26}.
It lifts the Yannakakis algorithm from \cqs to \crpqs and also relies on the output-sensitive \rpq algorithm from~\cite{rpq-pods25}.
See Appendix~\ref{app:comparison} for a detailed comparison with our work.

Recent work~\cite{10.1145/3725237} studies the problem of {\em minimizing} a \crpq,
i.e., finding the smallest\footnote{This typically means finding a \crpq with the minimum number of atoms.} equivalent \crpq. This problem is not data-dependent, but it depends
on the specific regular expressions used in the \rpq atoms of the \crpq.
In contrast, our results don't depend on the specific regular expressions.
Nevertheless, minimizing a \crpq can result in a smaller value of the free-connex fractional hypertree width, thus speeding up our algorithm.

The best known output-sensitive {\em combinatorial} algorithm for any acyclic \cq was recently proposed in~\cite{xiao-os-yannakakis,paris-os-yannakakis}.
However, certain \cqs are known to enjoy faster runtimes using non-combinatorial algorithms that rely on fast matrix multiplication.
For example, Boolean matrix multiplication (which corresponds to $Q_{k\text{-path}}$ for $k=2$) can be solved in time
$O(N^{2/3}\cdot\out^{2/3}+N\cdot\out^{1/3})$ {\em assuming $\omega = 2$}~\cite{DBLP:conf/icdt/AmossenP09}, where
$\omega$ is the matrix multiplication exponent~\cite{doi:10.1137/1.9781611977912.134}.
Recent work improved this bound to $O(N\cdot\out^{1/3}+\out)$~\cite{abboud2023time}.
The triangle query is another example on the cyclic side~\cite{BjorklundPWZ14}.
See~\cite{DBLP:conf/sigmod/DeepHK20} and~\cite{omega-subw} for further discussion on the use of fast matrix multiplication for CQ evaluation.

%% file: conclusion.tex
\section{Conclusion and Future Directions}
\label{sec:conclusion}

We presented an output sensitive algorithm for acyclic \crpqs that improves upon prior work.
Moreover, despite \crpqs seemingly being more difficult than \cqs due to their recursive nature,
we show that our algorithm already matches the runtime of the best known output-sensitive combinatorial algorithm for the corresponding \cqs.
Hence, any further improvement on our algorithm can only happen modulo an improvement in the state-of-the-art on \cq evaluation.

Our result opens up several interesting directions, some of which are immediate extensions of our work,
while others are more ambitious. We list some of them below:

\paragraph*{\bf Hybrid CQ/CRPQ queries.}
    In practice, we might encounter conjunctive queries that are {\em hybrid} between \crpqs and \cqs.
    In such queries, some atoms are \rpqs $\regpath{R}(X, Y)$ while others are relational atoms
    $R(\bm X)$ of arbitrary arities (not necessarily 2).
    Our results seem to {\em immediately} extend to such queries.
    To that end, we extend the definition of free-connex fractional hypertree width
    to those queries in the natural way.
    Since hybrid queries generalize \cqs, the generalized algorithm would also subsume
    the one from~\cite{xiao-os-yannakakis,paris-os-yannakakis}.

\paragraph*{\bf Aggregate \rpqs and \crpqs.}
    We can extend the semantics of \rpqs and \crpqs to {\em weighted} edge-labeled graphs
    in a natural way as follows:
    Consider an input edge-labeled graph where each edge has a weight, e.g., a non-negative real number. Given an \rpq $\regpath{R}(X, Y)$,
    can we list all pairs of vertices $x$ and $y$ such that there is a path from $x$ to $y$ whose label matches $\regpath{R}$, {\em and} also list the {\em minimum total weight} of such a path from $x$ to $y$ (i.e., the shortest path length)?
    To make the problem more general, we can choose any semiring $\bm K$ and
    consider a {\em $\bm K$-annotated} edge-labeled graph, i.e., an edge-labeled graph where each
    edge is annotated with an element from $\bm K$.
    Given a path in such graph, we define its annotation as the product of the annotations of its edges
    (using the multiplication operation of $\bm K$).
    Then, the semantics of an \rpq $\regpath{R}(X, Y)$ over such a graph is as follows: For each pair of vertices $x$ and $y$, compute a sum (using the addition operation of $\bm K$)
    of the annotations of all paths from $x$ to $y$ whose label is in the language defined by $\regpath{R}$.
    This semantics naturally extends to \crpqs as well,
    and also mirrors nicely the traditional semantics of aggregate \cqs over $\bm K$-annotated
    databases~\cite{provenance-semirings}.
    In the \cq world, prior output-sensitive results~\cite{xiao-os-yannakakis,paris-os-yannakakis}
    work for aggregate \cqs over $\bm K$-annotated databases out-of-the-box.
    But now we ask: can we extend our result to aggregate \crpqs over $\bm K$-annotated edge-labeled graphs?
    The answer doesn't seem to be as straightforward as in the \cq case,
    and it also seems to depend on the choice of the semiring $\bm K$:
    \begin{itemize}[leftmargin=*]
        \item For the tropical semiring $(\min, +)$, this seems to be possible.
        The main change in our algorithm that needs to be done is as follows: Currently, Proposition~\ref{prop:restriction-propagation}
        uses a breadth-first search; we need to replace it with Dijkstra's algorithm
        to compute shortest paths.
        \item For the $(+,\times)$ semiring, the problem can be a lot more challenging. In particular, it subsumes counting the number of paths in a graph, where these paths are of {\em unbounded} length. This is in contrast to the \cq world, where all paths have constant length in data complexity.
        It seems unlikely that we can still achieve the same runtime bound under this  semiring.
    \end{itemize}
    We leave characterizing the general {\em semiring-dependent} complexity of aggregate \crpqs as an open problem for future work.

\paragraph*{\bf Using Fast Matrix Multiplication to speed up \crpq evaluation.}
On the acyclic \cq side, certain queries admit non-combinatorial algorithms
that utilize fast matrix multiplication to beat the runtimes from~\cite{xiao-os-yannakakis,paris-os-yannakakis}.
Among those queries is the Boolean matrix multiplication query~\cite{DBLP:conf/icdt/AmossenP09,abboud2023time}, which corresponds to $Q_{k\text{-path}}$ for $k=2$ (Eq.~\eqref{eq:intro-k-path}).
It would be interesting to see if similar techniques can be leveraged to improve the output-sensitive runtimes for acyclic \crpqs as well.

\paragraph*{\bf Cyclic \crpqs.}
    Cyclic \crpqs seems to be far more challenging than their cyclic \cq counterparts.
    For example, consider the (Boolean) triangle \crpq:
    $Q() = \regpath{R_1}(X, Y) \wedge \regpath{R_2}(Y, Z) \wedge \regpath{R_3}(Z, X)$.
    While it might be tempting at first to try to mirror the success of Worst-case Optimal Join Algorithms for cyclic \cqs~\cite{Ngo:JACM:18,LeapFrogTrieJoin2014,WCOJGemsOfPODS2018,SkewStrikesBack2014},
    this seems to be a lot more challenging for \crpqs.
    This is because the $k$-expansion (Definition~\ref{defn:k-expansion}) of the triangle is a $k'$-cycle where $k'\defeq 3k$.
    In turn, as $k'$ goes to infinity, the best known complexity for the $k'$-cycle CQ approaches
    $O(N^2)$~\cite{AlonYZ95,theoretics:13722}, which is not a very interesting complexity for the triangle \crpq.
    We leave characterizing the complexity of cyclic \crpqs, including proving lower bounds, as an open question.

%% file: app_prelims.tex
\section{Missing details in Section~\ref{sec:prelims}}
\label{app:prelims}

\subsection{Proof of Proposition~\ref{prop:linear-rpq}} 
\PropLinearRPQ*
\begin{proof}
Assume that the CRPQ $Q$ is defined over an alphabet $\Sigma$,
and let $G = (V , E, \Sigma)$ be an input edge-labeled graph for $Q$.
Let $M_{\regpath{R}}=(V_{\regpath{R}}, E_{\regpath{R}}, \Sigma)$ be an edge-labeled graph representing the nondeterministic finite automaton (NFA) for the regular expression $\regpath{R}$.
Let $G'=(V', E')$ be the product graph of $G$ and $M_{\regpath{R}}$ (Definition~\ref{defn:product-graph}).
Recall that $V' = V \times V_{\regpath{R}}$.
Let $T\subseteq V'$ be the set of all vertices in $V'$ of the form $(v, q_f)$ where $v \in V$ and $q_f$ is an accepting state of $M_{\regpath{R}}$.
Similarly, let $S$ be the set of all vertices in $V'$ of the form $(u, q_0)$ where $u \in V$ and $q_0$ is the initial state of $M_{\regpath{R}}$.
To compute the answer of $Q$ on $G$,
we do a {\em backward} traversal of $G'$ starting from all vertices in $T$,
until we compute the subset $S' \subseteq S$ of vertices in $S$ reachable from $T$.
The answer of $Q$ on $G$ is then obtained by projecting each vertex in $S'$ onto its first component.

Since the size of $Q$, and hence of $\regpath{R}$, is constant,
the product graph $G'$ can be constructed in 
$O(|V| + |E|) = O(N)$ time.
Traversal takes $O(N)$ time as well.
\end{proof}

\subsection{Proof of Proposition~\ref{prop:filter-rpq}} 
\PropFilterRPQ*

\begin{proof}
Let $\sigma'$ be a fresh symbol not in $\Sigma$, and $\Sigma' \defeq \Sigma \cup \{\sigma'\}$.
We construct $G'=(V, E', \Sigma')$ by extending the edge set $E$ with self-loops at every vertex $s \in S$ with label $\sigma'$.
Now, let the regular expression $\regpath{R'}$ be the concatenation of the regular expression $\regpath{R}$ and the symbol $\sigma'$.
Then, answering the \rpq $\regpath{R'}(X, Y)$ over $G'$ is equivalent to answering the original query $Q(X, Y)$ over $G$.
\end{proof}

\subsection{Proof of Proposition~\ref{prop:k-expansion-fnfhtw}}
\PropKExpansionFnfhtw*
\begin{proof}
    Let $Q$ be an acyclic \crpq and $k \geq 2$ be any natural number.
    Let $Q'$ be the $k$-expansion of $Q$ (Definition~\ref{defn:k-expansion}).
    If $Q$ is trivial, then it is free-connex, hence $Q'$ is also free-connex.
    By Proposition~\ref{prop:fn-fhtw-properties}, $\fnfhtw(Q') = 1$.

    Now, we consider the case where $Q$ is non-trivial.
    Let $Q_1, \ldots, Q_m$ be the \boundconnected components of $Q$ (Definition~\ref{defn:bound-connected-components}).
    Since $Q$ is acyclic, each $Q_i$ is also acyclic.
    Moreover, let $Q_1', \ldots, Q_m'$be the $k$-expansions of $Q_1, \ldots, Q_m$, respectively.
    It is straightforward to verify that $Q_1', \ldots, Q_m'$ are the \boundconnected components of $Q'$.

    \begin{claim}
        Let $Q$ be a acyclic \boundconnected  \crpq (or acyclic \boundconnected  \cq with only binary atoms).
        Let $Q'$ be its $k$-expansion for some $k \geq 2$.
        We call $Q$ {\em single-edge} if it consists of a single atom $\regpath{R}(X, Y)$ where both $X$ and $Y$ are free variables. It must hold that
        \begin{align}
            \rho^*_{Q}(\free(Q)) &= \begin{cases}
                1 & \text{if $Q$ is single-edge}\\
                |\free(Q)| & \text{otherwise}
            \end{cases}
            \label{eq:rho-star-Q}\\
            \rho^*_{Q'}(\free(Q')) &= |\free(Q)|
            \label{eq:rho-star-Q'}
        \end{align}
        \label{claim:bound-connected-free-rho-star}
    \end{claim}
    Assume for now that the above claim holds.
    \begin{itemize}
    \item If no $Q_i$ is a single-edge, then $\rho^*_{Q_i}(\free(Q_i)) = \rho^*_{Q_i'}(\free(Q_i'))$ for every $i \in [m]$.
    Lemma~\ref{lem:fhtw_con_comp} implies that $\fnfhtw(Q) = \fnfhtw(Q')$.
    Moreover, since $Q$ is non-trivial, one $Q_i$ must contain at least two free variables.
    Hence, $\fnfhtw(Q) \geq 2$ and $\fnfhtw(Q') = \max(\fnfhtw(Q), 2)$ holds as well.
    \item If some $Q_i$ is a single-edge, then $\rho^*_{Q_i}(\free(Q_i)) = 1$ and $\rho^*_{Q_i'}(\free(Q_i')) = 2$.
        Therefore, by Lemma~\ref{lem:fhtw_con_comp}, $\fnfhtw(Q') =$ $\max(\fnfhtw(Q), 2)$.
    \end{itemize}

    \begin{proof}[Proof of Claim~\ref{claim:bound-connected-free-rho-star}]
    The above claim is proved similarly to Proposition~\ref{prop:number_free_rho_star}.
    In particular, because $Q$ is acyclic, \boundconnected, where all atoms are binary,
    $Q$ must be a tree where free variables can occur only at the leaves.
    The only way for a single edge to cover two free variables is when two leaves are directly connected by an edge,
    which can only happen when the query is a single edge.
    In all other cases, each free variable must be covered by a distinct edge.
    This proves Eq.~\eqref{eq:rho-star-Q}.

    To prove Eq.~\eqref{eq:rho-star-Q'}, note that $Q'$ must also be acyclic, \boundconnected, where all atoms are binary, and  $\free(Q') = \free(Q)$.
    Moreover, $Q'$ cannot be single-edge, even when $Q$ is single-edge. (Recall that $k\geq 2$.)
    \end{proof}
\end{proof}

%% file: app_results.tex
\section{Missing details in Section~\ref{sec:results}}
\label{app:results}

\subsection{Proof of Lemma~\ref{lem:fhtw_con_comp}}

\LemFhtwConComp*

Below, we only prove the $\geq$ direction in Eq.~\eqref{eq:fhtw-con-comp}
because our main result, Theorem~\ref{thm:main}, only uses this direction.
The $\leq$ direction is exclusively used in the proof of Proposition~\ref{prop:k-expansion-fnfhtw}, and we refer the reader to~\cite{xiao-os-yannakakis} (Lemmas 3.3 and 3.4) for a proof of that direction.
\begin{proof}[Proof of $\geq$ direction in Lemma~\ref{lem:fhtw_con_comp}]
We start with some definitions.
A {\em (multi-)hypergraph} $\calG = (\calV, \calE)$ consists of a finite set of vertices $\calV$ and a (multi-)set  of edges $\calE$ such that each edge is a subset of $\calV$ (of arbitrary size). 
Given a hypergraph $\calG = (\calV, \calE)$ and a vertex $v \in \calV$, we define the sets $\calE_v = \{e \in \calE \mid v \in e\}$ and 
$\calV_v = \bigcup_{e \in \calE_v} e$.
The hypergraph obtained by the elimination of $v$ in $\calG$ is defined as $\eliminate(\calG , v) = (\calV', \calE')$, where $\calV' = \calV \setminus \{v\}$ and $\calE' = (\calE \setminus \calE_v) \cup \{\calV_v \setminus \{v\}\}$.

Next, we recall the definition of variable elimination orders from prior work~\cite{faq,DBLP:journals/tods/KhamisCMNNOS20}. 
\begin{definition}[Variable Elimination Orders] 
 A  {\em variable elimination order} for a \crpq $Q$ is an ordering 
$\sigma = (X_1, \ldots, X_n)$ of its variables.
The {\em hypergraph sequence induced} by $\sigma$ is of the form 
$(\calG_0,  \calG_1, \ldots , \calG_n)$, where 
$\calG_0$ is the query graph of $Q$ and
$\calG_i = \eliminate(\calG_{i-1} , X_i)$ for $i \in [n]$.\footnote{In the original formulation, the elimination proceeds from the last variable to the first in the elimination order, that is, from right to left~\cite{faq}. For notational convenience, we instead traverse the elimination order from left to right.}
The {\em variable set sequence} induced by $\sigma$ is 
$(\calV_{X_1}, \ldots , \calV_{X_n})$. 
A variable elimination order is called {\em bound-leading} if no free variable appears before any bound one. 
\end{definition}

Prior work shows:
\begin{lemma}[\cite{faq,DBLP:journals/tods/KhamisCMNNOS20}]
\label{lem:elimination_order_tree-decomposition}
Let $Q$ be a \crpq.
For any free-connex tree decomposition $(T, \chi)$ for $Q$, there is a
bound-leading elimination order for $Q$ such that every set in the 
variable set sequence induced by $\sigma$ is contained in a bag $\chi(t)$
for some $t \in \nodes(T)$.
\end{lemma}

To prove Lemma~\ref{lem:fhtw_con_comp}, we use the following two auxiliary lemmas.

\begin{lemma}
\label{lem:elimination_order_icludes_components}   
    Let $Q$ be an acyclic \crpq with \boundconnected components $Q_1, \ldots, Q_k$ and let $(X_1, \ldots , X_n)$ be a bound-leading variable elimination 
    order for $Q$ with induced variable set sequence $(\calV_{X_1}, \ldots, \calV_{X_n})$. For any $i \in [k]$, there is a
    $j \in [n]$ such that $\free(Q_i) \subseteq \calV_{X_j}$. 
\end{lemma}
\begin{proof}
Consider an acyclic \crpq $Q$. Let $(X_1, \ldots , X_n)$ be a bound-leading variable elimination 
    order for $Q$ with induced variable set sequence $(\calV_{X_1}, \ldots, \calV_{X_n})$ and induced hypergraph sequence 
    $(\calG_0, \calG_1, \ldots , \calG_n)$. 
    Let $(X_1, \ldots , X_m)$ be the bound variables in $Q$.
    The main observation is that for any bound variable $X_j$, it holds:
    if $\calV_{X_j}$ contains two free variables, then they must be contained in the same \boundconnected component of $Q$.
    Consider the hypergraph $\calG_m$, which we obtain after eliminating all bound variables. Each hyperedge in $\calG_m$ is a set in the induced variable set sequence and  must consist of the free variables of a \boundconnected component of $Q$. Hence the free variables of each \boundconnected component of $Q$ must be contained in a set of the   induced variable set sequence. 
\end{proof}

\begin{lemma}
\label{lem:edge_cover_component}   
Let $Q$ be an acyclic \crpq and $Q'$ a \boundconnected component of $Q$. It holds
$\rho^*_{Q}(\free(Q')) = \rho^*_{Q'}(\free(Q'))$. 
\end{lemma}

\begin{proof}
Let $Q$ be an acyclic CRPQ and $Q'$ a \boundconnected component of $Q$.
Since $\atoms(Q) \supseteq \atoms(Q')$, it holds $\rho^*_{Q}(\free(Q')) \leq \rho^*_{Q'}(\free(Q'))$. Hence, it suffices to show that
$\rho^*_{Q}(\free(Q')) \geq \rho^*_{Q'}(\free(Q'))$.
For the sake of contradiction, assume that 
$\rho^*_{Q}(\free(Q')) < \rho^*_{Q'}(\free(Q'))$.
This means that $Q$ contains an atom $\regpath{R}(X, Y)$ such that $\regpath{R}(X, Y)$ is not contained in $Q'$ and $X$ and $Y$
are two distinct free variables in $Q'$. Since each pair of variables in $Q'$ is connected, this implies that $Q$ is cyclic, which is a contradiction to our initial assumption. 
\end{proof}


We are ready to prove Lemma~\ref{lem:fhtw_con_comp}.
Consider an acyclic \crpq $Q$ with \boundconnected components $Q_1, \ldots, Q_k$ and query graph $\calG = (\calV, \calE)$.
Consider a free-connex tree decomposition $(T, \chi)$ 
for $Q$ that is {\em optimal}, i.e., that satisfies
$\max_{t \in \nodes(T)} \rho^*_{Q}(\chi(t)) = \fnfhtw(Q)$. By Lemma~\ref{lem:elimination_order_tree-decomposition}, there exists  a bound-leading elimination order $(X_1, \ldots , X_n)$ for $Q$ with induced variable set sequence $(\calV_{X_1}, \ldots, \calV_{X_n})$ such that for every $j \in [n]$, there is $t \in \nodes(T)$
with $\calV_{X_j} \subseteq \chi(t)$. This implies  that 
$\fnfhtw(Q) \geq \max_{j \in [n]} \rho^*_{Q}(\calV_{X_j})$.
By Lemma~\ref{lem:elimination_order_icludes_components}, it holds that for every $i \in [k]$, there is a $j \in [n]$ such that $\free(Q_i) \subseteq \calV_{X_j}$. This implies that 
$\fnfhtw(Q) \geq \max_{i \in [k]} \rho^*_{Q}(\free(Q_i))$.
By Lemma~\ref{lem:edge_cover_component}, it holds  $\rho^*_{Q}(\free(Q_i)) = \rho^*_{Q_i}(\free(Q_i))$ for $i \in [k]$.
This means 
$ \fnfhtw(Q) \geq \max_{i \in [k]} \rho^*_{Q_i}(\free(Q_i))$.
Moreover, $\fnfhtw(Q) \geq 1$ by definition.
This proves the $\geq$ direction in Eq.~\eqref{eq:fhtw-con-comp}.
\end{proof}

\subsection{Proof of Proposition~\ref{prop:number_free_rho_star}}
\PropNumberFreeRhoStar*
\begin{proof}
    Since $Q$ is acyclic, \boundconnected, where all atoms are binary,
    $Q$ must be a tree where free variables can occur only at the leaves.
    In general, we need a distinct edge to cover each free leaf, i.e., $\rho^*_Q(\free(Q)) = |\free(Q)|$, {\em unless} there are two free leaves that are directly connected by an edge,
    which can only happen when the query consists of a single edge. In this case,
    we have $\rho^*_Q(\free(Q)) = 1$ and $|\free(Q)| = 2$.
\end{proof}

\subsection{Proof of Proposition~\ref{prop:restriction-propagation}}

\PropRestrictionPropagation*

\begin{proof}
    Let $M_{\regpath{R}}=(V_{\regpath{R}}, E_{\regpath{R}}, \Sigma)$ be an edge-labeled graph representing the NFA for the regular expression $\regpath{R}$.
    Let $G'=(V', E')$ be the product graph of $G$ and $M_{\regpath{R}}$ (Definition~\ref{defn:product-graph}).
    For each vertex $(v, q) \in (V' \defeq V\times V_{\regpath{R}})$ in $G'$, we compute a list, $\mylist(v, q)$, of up to $\Delta$ distinct $\bm Z$-tuples $\bm z$ such that
    there is a path in $G'$ from $(v, q)$ to some $(y, q_f)\in V'$ where $q_f$ is an accepting state of $M_{\regpath{R}}$
    and $(y, \bm z) \in \ov S$.
    This computation can be done in time $O(N \cdot \Delta)$ as follows:
    \begin{itemize}
        \item For each vertex $(y, q_f) \in V'$ where $q_f$ is an accepting state of $M_{\regpath{R}}$,
        we initialize $\mylist(y, q_f)$ to be the set of all $\bm Z$-tuples $\bm z$ such that $(y, \bm z) \in \ov S$.
        \item Every time we add a $\bm Z$-tuple $\bm z$ to $\mylist(v, q)$ for some vertex $(v, q) \in V'$, we go through all incoming edges $((v_2, q_2), (v, q))$ and we add $\bm z$ to $\mylist(v_2, q_2)$
        {\em assuming} $\bm z$ is not already in this list and $|\mylist(v_2, q_2)|<\Delta$.
        We repeat this process until no more tuples can be added to any list.
    \end{itemize}
    
    During the above process, each edge $((v_2, q_2), (v, q)) \in E'$ is traversed once every time
    a tuple $\bm z$ is added to $\mylist(v, q)$, and there can only be at most $\Delta$ such tuples.
    Hence, the overall runtime is $O(N \cdot \Delta)$.
    Finally, we set $\ov Q(X, \bm Z)$ as follows, where $q_0$ is the initial state of $M_{\regpath{R}}$:
    \begin{align*}
        \ov Q(X, \bm Z) \quad\defeq\quad \setof{(x, \bm z)}{x \in V \text{ and }\bm z \in \mylist(x, q_0)}
    \end{align*}

\end{proof}

\subsection{Proof of Lemma~\ref{lmm:free-leaf}}
\label{app:results:free-leaf}

\LmmFreeLeaf*

\begin{proof}
    Let $Q$ be any \freeleaf \crpq with $\ell$ leaves, i.e., $\ell\defeq |\free(Q)|$,
    and let $G$ be an input edge-labeled graph of size $N$.
    We assume for now that the output size $\out$ is known upfront; we will later explain how to get rid of this assumption.
    Let $\Delta \defeq \out^{1-1/\ell}$, and $\Delta' \defeq \Delta + 1$.
    Pick any free variable $W\in\free(Q)$, and re-orient the query tree so that $W$ becomes the root.
    For each variable $Y\in\vars(Q)$, let $\bm F_Y \subseteq \free(Q)$ be the set of free variables
    in the subtree rooted at $Y$ excluding $W$ (which can happen if $Y$ is the root $W$).
    We do a bottom-up traversal of the query tree, from the leaves to the root, $W$,
    where for each variable $Y$, let $X$ be its parent variable (assuming $Y$ is not the root), and let $\regpath{R}(X, Y)$ be the \rpq
    associated with the edge $(X, Y)$ in the query graph.
    For each variable $Y$, we define subqueries {\em without} actually computing them yet, as follows:
    \begin{itemize}
        \item If $Y$ is a leaf, we define the following subquery:
        (Recall that $Y$ must be a free variable in this case because $Q$ is \freeleaf.)
        \begin{align*}
            T_Y(X, Y) &\quad\defeq\quad \regpath{R}(X, Y)
                &\quad\text{(Prop.~\ref{prop:restriction-propagation:basecase})}
        \end{align*}
        \item Otherwise, Let $Z_1, \ldots, Z_k$ be the children of $Y$ in the query tree.
        We define the following two subqueries: (Note that in this case, $\bm F_{Z_1}, \ldots, \bm F_{Z_k}$ must be disjoint and their union must be $\bm F_Y$.)
        \begin{align*}
            S_Y(Y, \bm F_Y) &\quad\defeq\quad T_{Z_1}(Y, \bm F_{Z_1}) \wedge \cdots \wedge T_{Z_k}(Y, \bm F_{Z_k})
                \quad&\text{(Prop.~\ref{prop:restriction-composition})}&\\
            T_Y(X, \bm F_Y) &\quad\defeq\quad \regpath{R}(X, Y) \wedge S_Y(Y, \bm F_Y)
                \quad&\text{(Prop.~\ref{prop:restriction-propagation})}&
        \end{align*}
    \end{itemize}
    By induction, it is straightforward to verify the following two claims:
    \begin{claim}
        For each variable $Y$, $T_Y(X, \bm F_Y)$ is equivalent to a subquery of the \crpq $Q$
        with free variables $\{X\} \cup \bm F_Y$ and with edges corresponding to the subtree
        rooted at $Y$ in addition to the edge $(X, Y)$.
    \end{claim}
    \begin{claim}
        For each non-leaf variable $Y$, $S_Y(Y, \bm F_Y)$
        is equivalent to a subquery of the \crpq $Q$ with free variables $\{Y\} \cup \bm F_Y$ and with edges corresponding to the subtree rooted at $Y$.
    \end{claim}
    As a result, $S_W(W, \bm F_W)$ is equivalent to the original query $Q$.
    (Note that $\{W\}\cup \bm F_W = \free(Q)$.)
    By applying Prop.~\ref{prop:restriction-propagation:basecase},
    we can compute an $(X,\Delta')$-restriction $\ov T_Y(X, Y)$ of $T_Y(X, Y)$
    for each leaf variable $Y$ in time $O(N \cdot \Delta')$.
    Then, by repeatedly alternating Props.~\ref{prop:restriction-composition}
    and~\ref{prop:restriction-propagation},
    we can compute a $(Y, \Delta')$-restriction $\ov S_Y(Y, \bm F_Y)$ of $S_Y(Y, \bm F_Y)$
    and an $(X, \Delta')$-restriction $\ov T_Y(X, \bm F_Y)$ of $T_Y(X, \bm F_Y)$
    for each non-leaf variable $Y$ in time $O(N \cdot \Delta')$.
    At the very end, we compute a $(W, \Delta')$-restriction $\ov Q(W, \bm F_W)$ of $Q(W, \bm F_W)$
    in time $O(N \cdot \Delta')$.
    A $W$-value $w$ is called {\em light} if $\deg_{\ov Q}(\bm F_W|W=w) \leq \Delta$,
    and {\em heavy} otherwise (i.e., if $\deg_{\ov Q}(\bm F_W|W=w) = \Delta'\defeq \Delta+1$).
    The relation $\ov Q$ already contains all output tuples where the $W$-value is light.
    Let $H_W$ be the set of all heavy $W$-values.
    The size of $H_W$ is at most $\out/\Delta = \out^{1/\ell}$.
    Our next target is to report output tuples where the $W$-value is heavy.
    Using Prop.~\ref{prop:filter-rpq}, we can ``merge'' the filter $H_W$ into any \rpq
    that has $W$ as an endpoint.

    We repeat the above process for each free variable in $\free(Q)$.
    Every time we repeat the process, there will be one extra free variable $W$ whose
    number of different values is at most $\out^{1/\ell}$.
    Finally, let $Z$ be the {\em last} free variable we pick.
    At this point, every other free variable has at most $\out^{1/\ell}$ different values.
    Let $\bm F_Z\defeq \free(Q)\setminus \{Z\}$.
    The total number of different $\bm F_Z$-tuples is at most
    $(\out^{1/\ell})^{\ell-1}= \Delta$.
    Hence, at this point, all remaining output tuples must have light $Z$-values,
    and we can report all of them during the last bottom-up traversal.

    The above algorithm was under the assumption that the output size $\out$ is known upfront.
    If this is not the case, then we start by guessing $\out = 1$
    and keep doubling our guess every time we discover that our guess was below the actual output size,
    similar to Example~\ref{ex:free-leaf}.
    The total runtime would be $\sum_{i=1, \ldots, \lceil\log \out\rceil} N \cdot 2^{(\ell-1)i/\ell}$.
    This is a geometric series that is dominated by its last term, which is $O(N \cdot \out^{(\ell-1)/\ell})$.
\end{proof}

%% file: app_runtime_comparison.tex
\section{Comparison with Existing Work}
\label{app:comparison}

In this section, we compare the performance of our algorithm with the performance of two prior approaches for evaluating CRPQs. We use the following metrics in our complexity expressions:
$N$ is the size of the edge-labeled graph;
$\out$ is the output size of the query; 
$\out_a$ is the maximal size of the (uncalibrated) output of any RPQ appearing in the query;
$\fnfhtw$ is the free-connex fractional hypertree width of the CRPQ;
$\cw$ is a width measure for CRPQs used to specify the runtime of the algorithm proposed in \cite{crpq-icdt26}.

The first approach, which we call here \ospg\!+\osyan, involves constructing an equivalent CQ $Q'$ by materializing the results of each RPQ atom of $Q$, then evaluating $Q'$ using standard approaches. The best known algorithm for materializing RPQs \cite{rpq-pods25}, dubbed by the authors as \ospg, runs in time $O(N + N \cdot \out_a^{1/2})$. $Q'$ is then evaluated using an output-sensitive refinement of Yannakakis \cite{xiao-os-yannakakis}, called here \osyan, which is the best known output-sensitive algorithm for acyclic CQs. \osyan evaluates the CQ $Q'$ in time $O(\out_a \cdot \out^{1-1/\fnfhtw(Q')} + \out)$. Therefore, the total running time of \ospg\!+\osyan is $O(N + N \cdot \out_a^{1/2} + \out_a \cdot \out^{1-1/\fnfhtw(Q)} + \out)$.

The second approach, which we call here \pc\!+\ospg, is the output-sensitive algorithm for evaluating acyclic CRPQs proposed in \cite{crpq-icdt26}. The algorithm first transforms the original query $Q$ into a free-connex acyclic CRPQ $Q'$ by contracting some RPQs and "promoting" remaining bound variables into free variables. In the next step, the algorithm calibrates the RPQs of $Q'$ with each other before invoking \ospg on each RPQ. This yields an output-sensitive algorithm for evaluating $Q$ whose runtime depends on the number of "promoted" variables of $Q$. This quantity is defined in \cite{crpq-icdt26} as a width measure for CRPQs and is termed the "contraction width" ($\cw$) of a CRPQ.

The following table states the evaluation times achieved by 
\ospg\!+\osyan~\cite{rpq-pods25, xiao-os-yannakakis}, \pc\!+\ospg~\cite{crpq-icdt26}, and our algorithm for general acyclic queries.

\renewcommand{\arraystretch}{1.5}
\begin{center}
\begin{tabular}{@{}c | c@{}}
Algorithm & Running time for acyclic CRPQs \\
\hline

\ospg\!+\osyan~\cite{rpq-pods25, xiao-os-yannakakis} & 
$N + N\hspace{-0.1em}\cdot\hspace{-0.1em} \out_a^{1/2} + \out_a \hspace{-0.1em} \cdot\hspace{-0.1em} \out^{1-1/\fnfhtw(Q)} + \out$ \\
\hline

\pc\!+\ospg~\cite{crpq-icdt26} & 
$N + N\hspace{-0.1em}\cdot \hspace{-0.1em}\out^{1/2} \cdot N^{\cw/2}+ \out \cdot N^{\cw}$ \\
\hline

Our algorithm &
$N + N \cdot \out^{1-1/\max(\fnfhtw(Q), 2)} + \out$
\end{tabular}
\end{center}

We compare these three algorithms with an example.

\begin{figure}[ht]
\centering
\begin{tikzpicture}
  \node (u0) at (0,0) {$u_{0}$};
    \node (v0) at (1.5,0) {$v_{0}$};
    
    \node (z1) at (4,0) {$z_{1}$};
    \node (z2) at (4,-1) {$z_{2}$};    
    
  \node (u1) at (0,-0.5) {$u_{1}$};
  \node (udots) at (0,-1.1) {$\vdots$};
  \node (un) at (0,-2) {$u_{n}$};

  \node (v) at (1.5,-1) {$v$};

  \node (w1) at (3,-0.5) {$w_{1}$};
  \node (wdots) at (3,-1.1) {$\vdots$};
  \node (wn) at (3,-2) {$w_{n}$};

  \draw[<-] (u1) to node[above] {$a$} (v);
  \draw[<-] (un) to node[above] {$a$} (v);

  \draw[<-] (v) to  node[above] {$a$} (w1);
  \draw[<-] (v) to node[below] {$a$} (wn);

    \draw[->] (u0) to node[above] {$a$} (v0);
        \draw[->] (v0) to node[above] {$a$} (w1);

        \draw[->] (z1) to node[above] {$b$} (w1);        
        \draw[->] (z2) to node[below] {$c$} (w1);                

\end{tikzpicture}
\caption{The edge-labeled graph used in Example \ref{ex:running-time-comparison}}
\Description{The edge-labeled graph used in Example \ref{ex:running-time-comparison}}
\label{fig:running-time-comparison-graphs}
\end{figure}
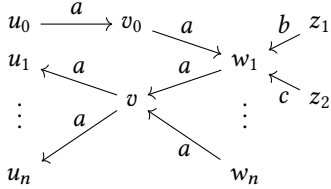

\begin{example}
 \label{ex:running-time-comparison}

Consider the 3-star CRPQ $Q(X_1,X_2, X_3) = a^*aa(X_1,X) \wedge b(X_2,X)\wedge c(X_3,X)$, which asks for all triples $(u,v,w)$ of vertices such that there is a vertex $z$, a path from $u$ to $z$ labeled by $a^*aa$, an edge from $v$ to $z$ labeled by $b$, and an edge from $w$ to $z$ labeled by $c$.
The query is acyclic but not free-connex. It has free-connex fractional hypertree width $w = 3$ and contraction width $\cw = 1$.

Consider the edge-labeled graph $G = (V_1 \cup V_2 \cup \{u_0, v_0, v, z_1, z_2\}, E_a \cup E_{bc}, \{a,b,c\})$, where $V_1 = \{u_1, \ldots , u_n\}$, $V_2 = \{w_1, \ldots , w_n\}$, 
$E_a = \{(w_i,a,v), (v,a,u_i) \mid i \in [n]\}$ $\cup$
$\{(u_0, a, v_0), (v_0, a, w_1)\}$, and $E_{bc} = \{(z_1,b,w_1), (z_2,c,w_1)\}$ 
for some $n \in \mathbb{N}$. The graph $G$ is visualized in Figure~\ref{fig:running-time-comparison-graphs} (left). 

We observe that  $|V| =2n+5$, $|E| = 2n + 4$, and $\out_a = \Theta(n^2)$.
The output of $Q$ consists of the single tuple $(u_0,z_1, z_2)$, hence $\out = 1$.
These quantities imply that \ospg\!+\osyan needs $O(n^2)$ time, \pc\!+\ospg needs $O(n^{3/2})$ time, while our algorithm needs only $O(n)$ time to evaluate the query. 
\qed
\end{example}

%% file: main.bbl

\begin{thebibliography}{37}


\ifx \showCODEN    \undefined \def \showCODEN     #1{\unskip}     \fi
\ifx \showDOI      \undefined \def \showDOI       #1{#1}\fi
\ifx \showISBNx    \undefined \def \showISBNx     #1{\unskip}     \fi
\ifx \showISBNxiii \undefined \def \showISBNxiii  #1{\unskip}     \fi
\ifx \showISSN     \undefined \def \showISSN      #1{\unskip}     \fi
\ifx \showLCCN     \undefined \def \showLCCN      #1{\unskip}     \fi
\ifx \shownote     \undefined \def \shownote      #1{#1}          \fi
\ifx \showarticletitle \undefined \def \showarticletitle #1{#1}   \fi
\ifx \showURL      \undefined \def \showURL       {\relax}        \fi
\providecommand\bibfield[2]{#2}
\providecommand\bibinfo[2]{#2}
\providecommand\natexlab[1]{#1}
\providecommand\showeprint[2][]{arXiv:#2}

\bibitem[Abboud et~al\mbox{.}(2024)]%
        {abboud2023time}
\bibfield{author}{\bibinfo{person}{Amir Abboud}, \bibinfo{person}{Karl Bringmann}, \bibinfo{person}{Nick Fischer}, {and} \bibinfo{person}{Marvin K{\"u}nnemann}.} \bibinfo{year}{2024}\natexlab{}.
\newblock \showarticletitle{The Time Complexity of Fully Sparse Matrix Multiplication}. In \bibinfo{booktitle}{\emph{Proceedings of the 2024 Annual ACM-SIAM Symposium on Discrete Algorithms (SODA)}}. SIAM, \bibinfo{pages}{4670--4703}.
\newblock


\bibitem[{Abo Khamis} et~al\mbox{.}(2020)]%
        {DBLP:journals/tods/KhamisCMNNOS20}
\bibfield{author}{\bibinfo{person}{Mahmoud {Abo Khamis}}, \bibinfo{person}{Ryan~R. Curtin}, \bibinfo{person}{Benjamin Moseley}, \bibinfo{person}{Hung~Q. Ngo}, \bibinfo{person}{XuanLong Nguyen}, \bibinfo{person}{Dan Olteanu}, {and} \bibinfo{person}{Maximilian Schleich}.} \bibinfo{year}{2020}\natexlab{}.
\newblock \showarticletitle{Functional Aggregate Queries with Additive Inequalities}.
\newblock \bibinfo{journal}{\emph{{ACM} Trans. Database Syst.}} \bibinfo{volume}{45}, \bibinfo{number}{4} (\bibinfo{year}{2020}), \bibinfo{pages}{17:1--17:41}.
\newblock
\urldef\tempurl%
\url{https://doi.org/10.1145/3426865}
\showDOI{\tempurl}


\bibitem[Abo~Khamis et~al\mbox{.}(2025)]%
        {omega-subw}
\bibfield{author}{\bibinfo{person}{Mahmoud Abo~Khamis}, \bibinfo{person}{Xiao Hu}, {and} \bibinfo{person}{Dan Suciu}.} \bibinfo{year}{2025}\natexlab{}.
\newblock \showarticletitle{Fast Matrix Multiplication meets the Submodular Width}.
\newblock \bibinfo{journal}{\emph{Proc. ACM Manag. Data}} \bibinfo{volume}{3}, \bibinfo{number}{2}, Article \bibinfo{articleno}{98} (\bibinfo{date}{June} \bibinfo{year}{2025}), \bibinfo{numpages}{26}~pages.
\newblock
\urldef\tempurl%
\url{https://doi.org/10.1145/3725235}
\showDOI{\tempurl}


\bibitem[{Abo Khamis} et~al\mbox{.}(2025)]%
        {crpq-icdt26}
\bibfield{author}{\bibinfo{person}{Mahmoud {Abo Khamis}}, \bibinfo{person}{Alexandru-Mihai {Hurjui}}, \bibinfo{person}{Ahmet {Kara}}, \bibinfo{person}{Dan {Olteanu}}, \bibinfo{person}{Dan {Suciu}}, {and} \bibinfo{person}{Zilu {Tian}}.} \bibinfo{year}{2025}\natexlab{}.
\newblock \showarticletitle{{Output-Sensitive Evaluation of Acyclic Conjunctive Regular Path Queries}}.
\newblock \bibinfo{journal}{\emph{arXiv e-prints}}, Article \bibinfo{articleno}{arXiv:2509.20204} (\bibinfo{date}{Sept.} \bibinfo{year}{2025}), \bibinfo{numpages}{arXiv:2509.20204}~pages.
\newblock
\urldef\tempurl%
\url{https://doi.org/10.48550/arXiv.2509.20204}
\showDOI{\tempurl}
\showeprint[arxiv]{2509.20204}~[cs.DB]


\bibitem[Abo~Khamis et~al\mbox{.}(2025)]%
        {rpq-pods25}
\bibfield{author}{\bibinfo{person}{Mahmoud Abo~Khamis}, \bibinfo{person}{Ahmet Kara}, \bibinfo{person}{Dan Olteanu}, {and} \bibinfo{person}{Dan Suciu}.} \bibinfo{year}{2025}\natexlab{}.
\newblock \showarticletitle{Output-Sensitive Evaluation of Regular Path Queries}.
\newblock  \bibinfo{volume}{3}, \bibinfo{number}{2}, Article \bibinfo{articleno}{105} (\bibinfo{date}{June} \bibinfo{year}{2025}), \bibinfo{numpages}{20}~pages.
\newblock
\urldef\tempurl%
\url{https://doi.org/10.1145/3725242}
\showDOI{\tempurl}


\bibitem[Abo~Khamis et~al\mbox{.}(2016)]%
        {faq}
\bibfield{author}{\bibinfo{person}{Mahmoud Abo~Khamis}, \bibinfo{person}{Hung~Q. Ngo}, {and} \bibinfo{person}{Atri Rudra}.} \bibinfo{year}{2016}\natexlab{}.
\newblock \showarticletitle{FAQ: Questions Asked Frequently}. In \bibinfo{booktitle}{\emph{{PODS}}}. \bibinfo{pages}{13–28}.
\newblock
\showISBNx{9781450341912}
\urldef\tempurl%
\url{https://doi.org/10.1145/2902251.2902280}
\showDOI{\tempurl}


\bibitem[{Abo Khamis} et~al\mbox{.}(2025)]%
        {theoretics:13722}
\bibfield{author}{\bibinfo{person}{Mahmoud {Abo Khamis}}, \bibinfo{person}{Hung~Q. Ngo}, {and} \bibinfo{person}{Dan Suciu}.} \bibinfo{year}{2025}\natexlab{}.
\newblock \showarticletitle{{PANDA: Query Evaluation in Submodular Width}}.
\newblock \bibinfo{journal}{\emph{TheoretiCS}}  \bibinfo{volume}{Volume 4}, Article \bibinfo{articleno}{12} (\bibinfo{date}{Apr} \bibinfo{year}{2025}).
\newblock
\showISSN{2751-4838}
\urldef\tempurl%
\url{https://doi.org/10.46298/theoretics.25.12}
\showDOI{\tempurl}


\bibitem[Alon et~al\mbox{.}(1995)]%
        {AlonYZ95}
\bibfield{author}{\bibinfo{person}{Noga Alon}, \bibinfo{person}{Raphael Yuster}, {and} \bibinfo{person}{Uri Zwick}.} \bibinfo{year}{1995}\natexlab{}.
\newblock \showarticletitle{Color-Coding}.
\newblock \bibinfo{journal}{\emph{J. {ACM}}} \bibinfo{volume}{42}, \bibinfo{number}{4} (\bibinfo{year}{1995}), \bibinfo{pages}{844--856}.
\newblock
\urldef\tempurl%
\url{https://doi.org/10.1145/210332.210337}
\showDOI{\tempurl}


\bibitem[Amossen and Pagh(2009)]%
        {DBLP:conf/icdt/AmossenP09}
\bibfield{author}{\bibinfo{person}{Rasmus~Resen Amossen} {and} \bibinfo{person}{Rasmus Pagh}.} \bibinfo{year}{2009}\natexlab{}.
\newblock \showarticletitle{Faster Join-Projects and Sparse Matrix Multiplications}. In \bibinfo{booktitle}{\emph{{ICDT}}}. \bibinfo{pages}{121--126}.
\newblock
\urldef\tempurl%
\url{https://doi.org/10.1145/1514894.1514909}
\showDOI{\tempurl}


\bibitem[Baeza(2013)]%
        {Baeza13}
\bibfield{author}{\bibinfo{person}{Pablo~Barcel{\'{o}} Baeza}.} \bibinfo{year}{2013}\natexlab{}.
\newblock \showarticletitle{Querying Graph Databases}. In \bibinfo{booktitle}{\emph{{PODS}}}. \bibinfo{pages}{175--188}.
\newblock
\urldef\tempurl%
\url{https://doi.org/10.1145/2463664.2465216}
\showDOI{\tempurl}


\bibitem[Bagan et~al\mbox{.}(2007)]%
        {BaganDG07}
\bibfield{author}{\bibinfo{person}{Guillaume Bagan}, \bibinfo{person}{Arnaud Durand}, {and} \bibinfo{person}{Etienne Grandjean}.} \bibinfo{year}{2007}\natexlab{}.
\newblock \showarticletitle{On Acyclic Conjunctive Queries and Constant Delay Enumeration}. In \bibinfo{booktitle}{\emph{{CSL}}}. \bibinfo{pages}{208--222}.
\newblock
\urldef\tempurl%
\url{https://doi.org/10.1007/978-3-540-74915-8\_18}
\showDOI{\tempurl}


\bibitem[Bj{\"{o}}rklund et~al\mbox{.}(2014)]%
        {BjorklundPWZ14}
\bibfield{author}{\bibinfo{person}{Andreas Bj{\"{o}}rklund}, \bibinfo{person}{Rasmus Pagh}, \bibinfo{person}{Virginia {Vassilevska Williams}}, {and} \bibinfo{person}{Uri Zwick}.} \bibinfo{year}{2014}\natexlab{}.
\newblock \showarticletitle{Listing Triangles}. In \bibinfo{booktitle}{\emph{{ICALP}}}. \bibinfo{pages}{223--234}.
\newblock
\urldef\tempurl%
\url{https://doi.org/10.1007/978-3-662-43948-7\_19}
\showDOI{\tempurl}


\bibitem[Bonifati et~al\mbox{.}(2017)]%
        {10.14778/3149193.3149196}
\bibfield{author}{\bibinfo{person}{Angela Bonifati}, \bibinfo{person}{Wim Martens}, {and} \bibinfo{person}{Thomas Timm}.} \bibinfo{year}{2017}\natexlab{}.
\newblock \showarticletitle{An analytical study of large SPARQL query logs}.
\newblock \bibinfo{journal}{\emph{Proc. VLDB Endow.}} \bibinfo{volume}{11}, \bibinfo{number}{2} (\bibinfo{date}{Oct.} \bibinfo{year}{2017}), \bibinfo{pages}{149–161}.
\newblock
\showISSN{2150-8097}
\urldef\tempurl%
\url{https://doi.org/10.14778/3149193.3149196}
\showDOI{\tempurl}


\bibitem[Bonifati et~al\mbox{.}(2019)]%
        {10.1145/3308558.3313472}
\bibfield{author}{\bibinfo{person}{Angela Bonifati}, \bibinfo{person}{Wim Martens}, {and} \bibinfo{person}{Thomas Timm}.} \bibinfo{year}{2019}\natexlab{}.
\newblock \showarticletitle{Navigating the Maze of Wikidata Query Logs}. In \bibinfo{booktitle}{\emph{The World Wide Web Conference}} (San Francisco, CA, USA) \emph{(\bibinfo{series}{WWW '19})}. \bibinfo{publisher}{Association for Computing Machinery}, \bibinfo{address}{New York, NY, USA}, \bibinfo{pages}{127–138}.
\newblock
\showISBNx{9781450366748}
\urldef\tempurl%
\url{https://doi.org/10.1145/3308558.3313472}
\showDOI{\tempurl}


\bibitem[Casel and Schmid(2023)]%
        {CaselS23}
\bibfield{author}{\bibinfo{person}{Katrin Casel} {and} \bibinfo{person}{Markus~L. Schmid}.} \bibinfo{year}{2023}\natexlab{}.
\newblock \showarticletitle{Fine-Grained Complexity of Regular Path Queries}.
\newblock \bibinfo{journal}{\emph{LMCS}} \bibinfo{volume}{19}, \bibinfo{number}{4} (\bibinfo{year}{2023}).
\newblock
\urldef\tempurl%
\url{https://doi.org/10.46298/LMCS-19(4:15)2023}
\showDOI{\tempurl}


\bibitem[Cucumides et~al\mbox{.}(2023)]%
        {DBLP:conf/icdt/CucumidesRV23}
\bibfield{author}{\bibinfo{person}{Tamara Cucumides}, \bibinfo{person}{Juan~L. Reutter}, {and} \bibinfo{person}{Domagoj Vrgoc}.} \bibinfo{year}{2023}\natexlab{}.
\newblock \showarticletitle{Size Bounds and Algorithms for Conjunctive Regular Path Queries}. In \bibinfo{booktitle}{\emph{{ICDT}}}. \bibinfo{pages}{13:1--13:17}.
\newblock
\urldef\tempurl%
\url{https://doi.org/10.4230/LIPICS.ICDT.2023.13}
\showDOI{\tempurl}


\bibitem[Deep et~al\mbox{.}(2020)]%
        {DBLP:conf/sigmod/DeepHK20}
\bibfield{author}{\bibinfo{person}{Shaleen Deep}, \bibinfo{person}{Xiao Hu}, {and} \bibinfo{person}{Paraschos Koutris}.} \bibinfo{year}{2020}\natexlab{}.
\newblock \showarticletitle{Fast Join Project Query Evaluation using Matrix Multiplication}. In \bibinfo{booktitle}{\emph{{SIGMOD}}}. \bibinfo{pages}{1213--1223}.
\newblock
\urldef\tempurl%
\url{https://doi.org/10.1145/3318464.3380607}
\showDOI{\tempurl}


\bibitem[Deep et~al\mbox{.}(2024)]%
        {paris-os-yannakakis}
\bibfield{author}{\bibinfo{person}{Shaleen Deep}, \bibinfo{person}{Hangdong Zhao}, \bibinfo{person}{Austen~Z. Fan}, {and} \bibinfo{person}{Paraschos Koutris}.} \bibinfo{year}{2024}\natexlab{}.
\newblock \showarticletitle{Output-sensitive Conjunctive Query Evaluation}.
\newblock  \bibinfo{volume}{2}, \bibinfo{number}{5}, Article \bibinfo{articleno}{220} (\bibinfo{date}{Nov.} \bibinfo{year}{2024}), \bibinfo{numpages}{24}~pages.
\newblock
\urldef\tempurl%
\url{https://doi.org/10.1145/3695838}
\showDOI{\tempurl}


\bibitem[Figueira et~al\mbox{.}(2020)]%
        {FigueiraGKMNT20}
\bibfield{author}{\bibinfo{person}{Diego Figueira}, \bibinfo{person}{Adwait Godbole}, \bibinfo{person}{S. Krishna}, \bibinfo{person}{Wim Martens}, \bibinfo{person}{Matthias Niewerth}, {and} \bibinfo{person}{Tina Trautner}.} \bibinfo{year}{2020}\natexlab{}.
\newblock \showarticletitle{Containment of Simple Conjunctive Regular Path Queries}. In \bibinfo{booktitle}{\emph{{KR}}}. \bibinfo{pages}{371--380}.
\newblock
\urldef\tempurl%
\url{https://doi.org/10.24963/KR.2020/38}
\showDOI{\tempurl}


\bibitem[Figueira et~al\mbox{.}(2025)]%
        {10.1145/3725237}
\bibfield{author}{\bibinfo{person}{Diego Figueira}, \bibinfo{person}{R\'{e}mi Morvan}, {and} \bibinfo{person}{Miguel Romero}.} \bibinfo{year}{2025}\natexlab{}.
\newblock \showarticletitle{Minimizing Conjunctive Regular Path Queries}.
\newblock \bibinfo{journal}{\emph{Proc. ACM Manag. Data}} \bibinfo{volume}{3}, \bibinfo{number}{2}, Article \bibinfo{articleno}{100} (\bibinfo{date}{June} \bibinfo{year}{2025}), \bibinfo{numpages}{25}~pages.
\newblock
\urldef\tempurl%
\url{https://doi.org/10.1145/3725237}
\showDOI{\tempurl}


\bibitem[Francis et~al\mbox{.}(2018)]%
        {Cypher:SIGMOD:2018}
\bibfield{author}{\bibinfo{person}{Nadime Francis}, \bibinfo{person}{Alastair Green}, \bibinfo{person}{Paolo Guagliardo}, \bibinfo{person}{Leonid Libkin}, \bibinfo{person}{Tobias Lindaaker}, \bibinfo{person}{Victor Marsault}, \bibinfo{person}{Stefan Plantikow}, \bibinfo{person}{Mats Rydberg}, \bibinfo{person}{Petra Selmer}, {and} \bibinfo{person}{Andr{\'{e}}s Taylor}.} \bibinfo{year}{2018}\natexlab{}.
\newblock \showarticletitle{Cypher: An Evolving Query Language for Property Graphs}. In \bibinfo{booktitle}{\emph{{SIGMOD}}}. \bibinfo{pages}{1433--1445}.
\newblock
\urldef\tempurl%
\url{https://doi.org/10.1145/3183713.3190657}
\showDOI{\tempurl}


\bibitem[Green et~al\mbox{.}(2007)]%
        {provenance-semirings}
\bibfield{author}{\bibinfo{person}{Todd~J. Green}, \bibinfo{person}{Grigoris Karvounarakis}, {and} \bibinfo{person}{Val Tannen}.} \bibinfo{year}{2007}\natexlab{}.
\newblock \showarticletitle{Provenance semirings}. In \bibinfo{booktitle}{\emph{Proceedings of the Twenty-Sixth ACM SIGMOD-SIGACT-SIGART Symposium on Principles of Database Systems}} (Beijing, China) \emph{(\bibinfo{series}{PODS '07})}. \bibinfo{publisher}{Association for Computing Machinery}, \bibinfo{address}{New York, NY, USA}, \bibinfo{pages}{31–40}.
\newblock
\showISBNx{9781595936851}
\urldef\tempurl%
\url{https://doi.org/10.1145/1265530.1265535}
\showDOI{\tempurl}


\bibitem[Harris and Seaborne(2013)]%
        {SPARQL:2013}
\bibfield{author}{\bibinfo{person}{Steve Harris} {and} \bibinfo{person}{Andy Seaborne}.} \bibinfo{year}{2013}\natexlab{}.
\newblock \bibinfo{title}{SPARQL 1.1 Query Language. W3C Recommendation}.
\newblock \bibinfo{howpublished}{\url{http://www.w3.org/TR/sparql11-query/}}.
\newblock


\bibitem[Hu(2025)]%
        {xiao-os-yannakakis}
\bibfield{author}{\bibinfo{person}{Xiao Hu}.} \bibinfo{year}{2025}\natexlab{}.
\newblock \showarticletitle{Output-Optimal Algorithms for Join-Aggregate Queries}.
\newblock \bibinfo{journal}{\emph{Proc. ACM Manag. Data}} \bibinfo{volume}{3}, \bibinfo{number}{2}, Article \bibinfo{articleno}{104} (\bibinfo{date}{June} \bibinfo{year}{2025}), \bibinfo{numpages}{27}~pages.
\newblock
\urldef\tempurl%
\url{https://doi.org/10.1145/3725241}
\showDOI{\tempurl}


\bibitem[Malyshev et~al\mbox{.}(2018)]%
        {10.1007/978-3-030-00668-6_23}
\bibfield{author}{\bibinfo{person}{Stanislav Malyshev}, \bibinfo{person}{Markus Kr\"{o}tzsch}, \bibinfo{person}{Larry Gonz\'{a}lez}, \bibinfo{person}{Julius Gonsior}, {and} \bibinfo{person}{Adrian Bielefeldt}.} \bibinfo{year}{2018}\natexlab{}.
\newblock \showarticletitle{Getting the Most Out of Wikidata: Semantic Technology Usage in Wikipedia’s Knowledge Graph}. In \bibinfo{booktitle}{\emph{The Semantic Web – ISWC 2018: 17th International Semantic Web Conference, Monterey, CA, USA, October 8–12, 2018, Proceedings, Part II}} (Monterey, CA, USA). \bibinfo{publisher}{Springer-Verlag}, \bibinfo{address}{Berlin, Heidelberg}, \bibinfo{pages}{376–394}.
\newblock
\showISBNx{978-3-030-00667-9}
\urldef\tempurl%
\url{https://doi.org/10.1007/978-3-030-00668-6_23}
\showDOI{\tempurl}


\bibitem[Martens and Trautner(2018)]%
        {MartensT18}
\bibfield{author}{\bibinfo{person}{Wim Martens} {and} \bibinfo{person}{Tina Trautner}.} \bibinfo{year}{2018}\natexlab{}.
\newblock \showarticletitle{Evaluation and Enumeration Problems for Regular Path Queries}. In \bibinfo{booktitle}{\emph{{ICDT}}}. \bibinfo{pages}{19:1--19:21}.
\newblock
\urldef\tempurl%
\url{https://doi.org/10.4230/LIPICS.ICDT.2018.19}
\showDOI{\tempurl}


\bibitem[Martens and Trautner(2019)]%
        {MartensT19b}
\bibfield{author}{\bibinfo{person}{Wim Martens} {and} \bibinfo{person}{Tina Trautner}.} \bibinfo{year}{2019}\natexlab{}.
\newblock \showarticletitle{Bridging Theory and Practice with Query Log Analysis}.
\newblock \bibinfo{journal}{\emph{{SIGMOD} Rec.}} \bibinfo{volume}{48}, \bibinfo{number}{1} (\bibinfo{year}{2019}), \bibinfo{pages}{6--13}.
\newblock
\urldef\tempurl%
\url{https://doi.org/10.1145/3371316.3371319}
\showDOI{\tempurl}


\bibitem[Mendelzon and Wood(1995)]%
        {MendelzonW95}
\bibfield{author}{\bibinfo{person}{Alberto~O. Mendelzon} {and} \bibinfo{person}{Peter~T. Wood}.} \bibinfo{year}{1995}\natexlab{}.
\newblock \showarticletitle{Finding Regular Simple Paths in Graph Databases}.
\newblock \bibinfo{journal}{\emph{{SIAM} J. Comput.}} \bibinfo{volume}{24}, \bibinfo{number}{6} (\bibinfo{year}{1995}), \bibinfo{pages}{1235--1258}.
\newblock
\urldef\tempurl%
\url{https://doi.org/10.1137/S009753979122370X}
\showDOI{\tempurl}


\bibitem[Ngo(2018)]%
        {WCOJGemsOfPODS2018}
\bibfield{author}{\bibinfo{person}{Hung~Q. Ngo}.} \bibinfo{year}{2018}\natexlab{}.
\newblock \showarticletitle{Worst-Case Optimal Join Algorithms: Techniques, Results, and Open Problems}. In \bibinfo{booktitle}{\emph{{PODS}}}. \bibinfo{pages}{111–124}.
\newblock
\showISBNx{9781450347068}
\urldef\tempurl%
\url{https://doi.org/10.1145/3196959.3196990}
\showDOI{\tempurl}


\bibitem[Ngo et~al\mbox{.}(2018)]%
        {Ngo:JACM:18}
\bibfield{author}{\bibinfo{person}{Hung~Q. Ngo}, \bibinfo{person}{Ely Porat}, \bibinfo{person}{Christopher R{\'{e}}}, {and} \bibinfo{person}{Atri Rudra}.} \bibinfo{year}{2018}\natexlab{}.
\newblock \showarticletitle{{Worst-case Optimal Join Algorithms}}.
\newblock \bibinfo{journal}{\emph{J. {ACM}}} \bibinfo{volume}{65}, \bibinfo{number}{3} (\bibinfo{year}{2018}), \bibinfo{pages}{16:1--16:40}.
\newblock
\urldef\tempurl%
\url{https://doi.org/10.1145/2213556.2213565}
\showDOI{\tempurl}


\bibitem[Ngo et~al\mbox{.}(2014)]%
        {SkewStrikesBack2014}
\bibfield{author}{\bibinfo{person}{Hung~Q Ngo}, \bibinfo{person}{Christopher R\'{e}}, {and} \bibinfo{person}{Atri Rudra}.} \bibinfo{year}{2014}\natexlab{}.
\newblock \showarticletitle{Skew Strikes Back: New Developments in the Theory of Join Algorithms}.
\newblock \bibinfo{journal}{\emph{SIGMOD Rec.}} \bibinfo{volume}{42}, \bibinfo{number}{4} (\bibinfo{date}{feb} \bibinfo{year}{2014}), \bibinfo{pages}{5–16}.
\newblock
\showISSN{0163-5808}
\urldef\tempurl%
\url{https://doi.org/10.1145/2590989.2590991}
\showDOI{\tempurl}


\bibitem[Olteanu and Z{\'{a}}vodn{\'{y}}(2015)]%
        {DBLP:journals/tods/OlteanuZ15}
\bibfield{author}{\bibinfo{person}{Dan Olteanu} {and} \bibinfo{person}{Jakub Z{\'{a}}vodn{\'{y}}}.} \bibinfo{year}{2015}\natexlab{}.
\newblock \showarticletitle{Size Bounds for Factorised Representations of Query Results}.
\newblock \bibinfo{journal}{\emph{{ACM} Trans. Database Syst.}} \bibinfo{volume}{40}, \bibinfo{number}{1} (\bibinfo{year}{2015}), \bibinfo{pages}{2:1--2:44}.
\newblock
\urldef\tempurl%
\url{https://doi.org/10.1145/2656335}
\showDOI{\tempurl}


\bibitem[Segoufin(2013)]%
        {10.1145/2448496.2448498}
\bibfield{author}{\bibinfo{person}{Luc Segoufin}.} \bibinfo{year}{2013}\natexlab{}.
\newblock \showarticletitle{Enumerating with constant delay the answers to a query}. In \bibinfo{booktitle}{\emph{Proceedings of the 16th International Conference on Database Theory}} (Genoa, Italy) \emph{(\bibinfo{series}{ICDT '13})}. \bibinfo{publisher}{Association for Computing Machinery}, \bibinfo{address}{New York, NY, USA}, \bibinfo{pages}{10–20}.
\newblock
\showISBNx{9781450315982}
\urldef\tempurl%
\url{https://doi.org/10.1145/2448496.2448498}
\showDOI{\tempurl}


\bibitem[Strassen(1969)]%
        {Strassen1969GaussianEI}
\bibfield{author}{\bibinfo{person}{Volker Strassen}.} \bibinfo{year}{1969}\natexlab{}.
\newblock \showarticletitle{Gaussian elimination is not optimal}.
\newblock \bibinfo{journal}{\emph{Numer. Math.}}  \bibinfo{volume}{13} (\bibinfo{year}{1969}), \bibinfo{pages}{354--356}.
\newblock
\urldef\tempurl%
\url{https://api.semanticscholar.org/CorpusID:121656251}
\showURL{%
\tempurl}


\bibitem[Veldhuizen(2014)]%
        {LeapFrogTrieJoin2014}
\bibfield{author}{\bibinfo{person}{Todd~L. Veldhuizen}.} \bibinfo{year}{2014}\natexlab{}.
\newblock \showarticletitle{Triejoin: {A} Simple, Worst-Case Optimal Join Algorithm}. In \bibinfo{booktitle}{\emph{{ICDT}}}. \bibinfo{pages}{96--106}.
\newblock
\urldef\tempurl%
\url{https://doi.org/10.5441/002/ICDT.2014.13}
\showDOI{\tempurl}


\bibitem[Williams et~al\mbox{.}({[n.\,d.]})]%
        {doi:10.1137/1.9781611977912.134}
\bibfield{author}{\bibinfo{person}{Virginia~Vassilevska Williams}, \bibinfo{person}{Yinzhan Xu}, \bibinfo{person}{Zixuan Xu}, {and} \bibinfo{person}{Renfei Zhou}.} \bibinfo{year}{[n.\,d.]}\natexlab{}.
\newblock \bibinfo{booktitle}{\emph{New Bounds for Matrix Multiplication: from Alpha to Omega}}.
\newblock \bibinfo{pages}{3792--3835}.
\newblock
\urldef\tempurl%
\url{https://doi.org/10.1137/1.9781611977912.134}
\showDOI{\tempurl}
\showeprint{https://epubs.siam.org/doi/pdf/10.1137/1.9781611977912.134}


\bibitem[Yannakakis(1981)]%
        {Yannakakis81}
\bibfield{author}{\bibinfo{person}{Mihalis Yannakakis}.} \bibinfo{year}{1981}\natexlab{}.
\newblock \showarticletitle{Algorithms for Acyclic Database Schemes}. In \bibinfo{booktitle}{\emph{VLDB}}. \bibinfo{pages}{82--94}.
\newblock


\end{thebibliography}
